\documentclass[a4paper,twocolumn,11pt,accepted=2025-10-14]{quantumarticle}
\pdfoutput=1

\usepackage[utf8]{inputenc}
\usepackage[english]{babel}
\usepackage[T1]{fontenc}
\usepackage{microtype}
\usepackage[all]{nowidow}

\usepackage{amsmath}
\usepackage{amssymb}
\usepackage{amsthm}
\usepackage{stmaryrd}
\usepackage{mathtools}
\usepackage{thmtools}
\usepackage{bm}
\usepackage{cancel}
\usepackage[ruled,vlined,linesnumbered]{algorithm2e}
\usepackage{multirow}

\usepackage{enumitem}

\usepackage{hyperref}
\usepackage[capitalize]{cleveref}

\usepackage{tikz}
\usetikzlibrary{%
  arrows.meta,
  positioning,
  decorations.markings,
  patterns
}
\usepackage{tikz-cd}
\usepackage{quiver}

\usepackage{tikzit}

\tikzstyle{bullet}=[fill={rgb,255: red,64; green,64; blue,64}, draw=black, shape=circle, inner sep=1pt, minimum size=3pt, text=white]
\tikzstyle{white-circle}=[fill=white, draw=black, shape=circle, inner sep=1pt, minimum size=3pt]
\tikzstyle{box}=[fill=white, draw=black, shape=rectangle]
\tikzstyle{wide copoint}=[fill=white, draw, shape=isosceles triangle, shape border rotate=90, inner sep=0pt, minimum width=1.5cm, minimum height=6.12mm]
\tikzstyle{wide point}=[fill=white, draw, shape=isosceles triangle, shape border rotate=-90, inner sep=0pt, minimum width=1.5cm, minimum height=6.12mm, yshift=-0.0mm]
\tikzstyle{black_diamond}=[fill=black, draw=black, shape=rectangle, rotate=45]
\tikzstyle{small_label}=[scale=0.7, shape=circle]
\tikzstyle{measurement}=[fill=white, draw=black, shape=semicircle, append after command={
        [draw, -{Latex[length=1.5mm]}] (\tikzlastnode.south) edge ++(60:0.9cm)
    }]
\tikzstyle{small label fill}=[fill=yellow, draw=none, shape=rectangle, scale=0.7, inner sep=0]
\tikzstyle{red small label}=[fill=none, draw=none, shape=circle, text=red, scale=0.7, tikzit draw=red, tikzit fill=white]
\tikzstyle{blue small label}=[fill=none, draw=none, shape=circle, text=blue, scale=0.7, tikzit fill=white, tikzit draw=blue]
\tikzstyle{ground}=[draw, shape=ground, minimum width=0.6cm, minimum height=0.4cm, inner sep=0pt]
\tikzstyle{red bullet}=[fill=red, draw=red, shape=circle, inner sep=0pt, minimum size=3pt]
\tikzstyle{blue bullet}=[fill=blue, draw=blue, shape=circle, inner sep=0pt, minimum size=3pt]
\tikzstyle{arrow}=[arrows={-Latex[width=5pt, length=5pt]}]
\tikzstyle{latexarrow}=[arrows={-Stealth[scale=1.5]}]

\tikzstyle{classical}=[-, fill=none, draw=black, double]
\tikzstyle{measure-line}=[{|-|}, line width=0.3pt]
\tikzstyle{dashed}=[-, dash pattern=on 2pt off 2pt, draw opacity=0.7]
\tikzstyle{thick}=[-, line width=1pt]
\tikzstyle{thick red}=[-, line width=1pt, draw=red]
\tikzstyle{thick blue}=[-, line width=1pt, draw=blue]
\tikzstyle{brace}=[-, decorate, decoration={brace, amplitude=10pt}]
\tikzstyle{yellow fill}=[-, fill=yellow, draw=none, tikzit draw={rgb,255: red,210; green,214; blue,0}]
\tikzstyle{normal with red glow}=[-, draw, glow=red]
\tikzstyle{mid arrow}=[-, draw=black, postaction=decorate, decoration={{
      markings,
      mark=at position 0.55 with {\arrow{Triangle[scale=0.8]}}}
}, tikzit draw={rgb,255: red,239; green,0; blue,247}]

\usepackage{ytableau}
\usepackage{pgf}

\usepackage{caption}
\usepackage{subcaption}

\newcommand{\red}[1]{{\color{red}#1}}
\newcommand{\blue}[1]{{\color{blue}#1}}

\counterwithin{equation}{section}

\declaretheoremstyle[
  headfont=\normalfont\bfseries,
  numberwithin=section,
  bodyfont=\normalfont,
  headformat={\NAME\ \NUMBER \NOTE}
]{mydefstyle}

\declaretheorem{definition}[style=mydefstyle, title={Definition}, refname={definition,definitions}, Refname={Definition,Definitions}]

\declaretheoremstyle[
  headfont=\normalfont\bfseries,
  sibling=definition,
  bodyfont=\normalfont,
  headformat={\NAME\ \NUMBER \NOTE}
]{mythmstyle}

\declaretheorem{theorem}[style=mythmstyle, title=Theorem, refname={theorem,theorems}, Refname={Theorem,Theorems}]
\declaretheorem{lemma}[style=mythmstyle, title=Lemma, refname={lemma,lemmas}, Refname={Lemma,Lemmas}]
\declaretheorem{corollary}[style=mythmstyle, title=Corollary, refname={corollary,corollaries}, Refname={Corollary,Corollaries}]

\renewcommand{\vec}[1]{\underline{#1}}
\newcommand{\LL}{\mathcal{L}}
\newcommand{\gen}[1]{\left\langle #1 \right\rangle}
\newcommand{\negen}[1]{\langle #1 \rangle}

\newcommand{\N}{\mathbb{N}}
\newcommand{\Z}{\mathbb{Z}}

\newcommand{\R}{\mathbb{R}}
\newcommand{\C}{\mathbb{C}}

\newcommand{\then}{\fatsemi}
\DeclareMathOperator*{\E}{\mathbb{E}}

\newcommand{\ket}[1]{\left|#1\right\rangle}
\newcommand{\fket}[1]{\ket{#1}}
\newcommand{\neket}[1]{|#1\rangle}

\newcommand{\deq}{\coloneqq}

\makeatletter
\newcommand{\oset}[3][0ex]{%
  \mathrel{\mathop{#3}\limits^{
    \vbox to#1{\kern-2\ex@
    \hbox{$\scriptstyle#2$}\vss}}}}
\makeatother
\newcommand{\isoto}{\oset[-0.3ex]{\!\!\sim}{\to}}

\makeatletter
\newcommand{\etal}{%
  \unskip\nobreak~\textit{et~al.}%
  \@ifnextchar\cite{\nobreak~}{\@\ }%
}
\makeatother

\newcommand{\iicc}[1]{\llbracket #1 \rrbracket}

\newcommand{\downset}[1]{{#1{\downarrow}}} 
\newcommand{\Fs}[2]{\mathcal{F}^{#1} \C^{#2}}

\newcommand{\fwtikzfig}[1]{%
  \resizebox{\linewidth}{!}{
    \input{#1.tikz}
  }%
}

\tikzset{glow/.style={%
  preaction={draw,
             line cap=round,
             line join=round,
             opacity=0.3,
             line width=4pt,
             #1}}}
\tikzset{glow/.default=yellow}

\DeclareRobustCommand{\tikzredarrowsback}{%
  \mathrel{\raisebox{0.5ex}{%
      \tikz[baseline=(current bounding box.center)]{
        \draw [style=none, red, glow={red!75},
        postaction={decorate,
          decoration={markings,
            mark=between positions 0.2 and 1 step 0.4 with {\arrow[scale=1.2]{stealth}}}}] (0,0) to ++(-0.5,0);
      }}}
}

\makeatletter
\newcommand{\removelatexerror}{\let\@latex@error\@gobble}
\makeatother

\newenvironment{myalgorithm}[2]{
  \begingroup
  \newcommand*{\mycaption}{#1}
  \newcommand*{\mypostfix}{#2}
  \begin{figure}[htb!]
    \begingroup
    \removelatexerror
    \begin{algorithm}[H]
      \DontPrintSemicolon
}{
    \end{algorithm}
    \endgroup
    \caption{\mycaption}
    \mypostfix
  \end{figure}
  \endgroup
}

\setlist[enumerate]{
  labelwidth=*,
  leftmargin=0pt,
  itemindent=\parindent+\labelwidth,
  itemsep=0pt,
  parsep=\parsep,
  listparindent=\parindent,
  topsep=\parsep
}

\setlist[itemize]{
  labelwidth=*,
  leftmargin=0pt,
  itemindent=\parindent+\labelwidth,
  itemsep=0pt,
  parsep=\parsep,
  listparindent=\parindent,
  topsep=\parsep
}

\newcommand{\captionitem}[1]{\hspace*{\parindent}\textbf{(#1)}}

\AddToHook{cmd/appendix/before}{%
  \crefalias{section}{appendix}%
  \crefalias{subsection}{appendix}%
  \crefalias{subsubsection}{appendix}
}

\begin{document}

\title{Boundaries for quantum advantage with single photons and loop-based time-bin interferometers}
\author{Samo Nov\'{a}k}
\orcid{0000-0002-2713-9593}
\affiliation{ORCA Computing, London, UK}

\author{David D. Roberts}
\orcid{0009-0009-1502-5845}
\affiliation{ORCA Computing, London, UK}

\author{Alexander Makarovskiy}
\affiliation{ORCA Computing, London, UK}

\author{Ra\'{u}l Garc\'{i}a-Patr\'{o}n}
\orcid{0000-0003-1760-433X}
\affiliation{%
  Laboratory for the Foundations of Computer Science,
  School of Informatics,
  University of Edinburgh
}
\affiliation{Phasecraft Ltd., London, UK}

\author{William R. Clements}
\email{wclements@orcacomputing.com}
\affiliation{ORCA Computing, London, UK}

\date{25 November 2024}
\begin{abstract}
  Loop-based boson samplers interfere photons in the time degree of freedom using a sequence of delay lines. Since they require few hardware components while also allowing for long-range entanglement, they are strong candidates for demonstrating quantum advantage beyond the reach of classical emulation.
  We propose a method to exploit this loop-based structure to more efficiently classically sample from such systems.
  Our algorithm exploits a causal-cone argument to decompose the circuit into smaller effective components that can each be simulated sequentially by calling a state vector simulator as a subroutine.
  To quantify the complexity of our approach, we develop a new lattice path formalism that allows us to efficiently characterize the state space that must be tracked during the simulation.
  In addition, we develop a heuristic method that allows us to predict the expected average and worst-case memory requirements of running these simulations.
  We use these methods to compare the simulation complexity of different families of loop-based interferometers,
  allowing us to quantify the potential for quantum advantage of single-photon Boson Sampling in loop-based architectures.
\end{abstract}


\maketitle

\section{Introduction}

\begin{figure*}[tp]
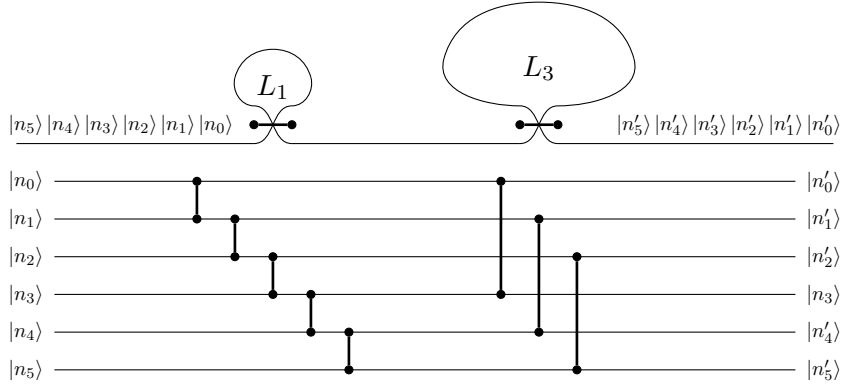
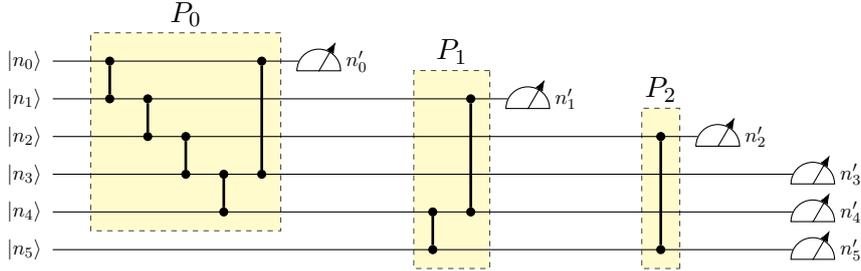
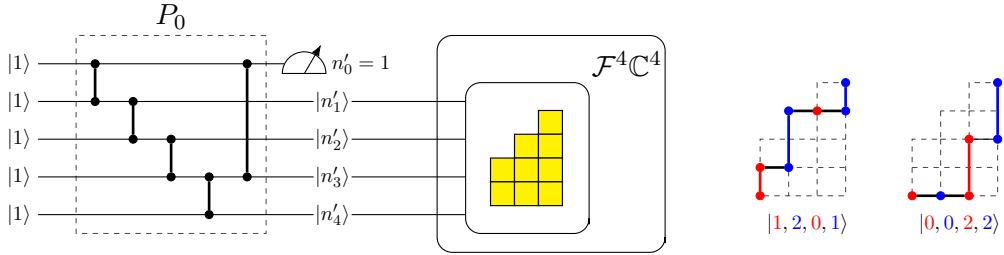

  \begin{subfigure}[b]{\textwidth}
    \ctikzfig{L1-3_circuit-finite}
    \caption{A system made of loops $L_1 \then L_3$}
    \label{fig:big_picture:circuit}
  \end{subfigure}

  \vspace*{2em}

  \begin{subfigure}[b]{1\textwidth}
    \ctikzfig{L1-3_progressive}
    \caption{Progressive decomposition of $L_1 \then L_3$}
    \label{fig:big_picture:progressive}
  \end{subfigure}

  \vspace*{2em}

  \begin{subfigure}[b]{1\textwidth}
    \ytableausetup{boxsize=3mm}
    \ctikzfig{L1-3_progressive-LP}
    \caption{State space of possible outputs after the zeroth component $P_0$ of $L_1 \then L_3$ and measurement $n'_0 = 1$}
    \label{fig:big_picture:lattice-paths}
  \end{subfigure}

  \caption{\textbf{Overview of the results.}
    \captionitem{a}~A time-bin boson sampler consisting of loops of different lengths composed in sequence. The input photons arrive in regular intervals, represented by time-bin modes with photon numbers $\ket{n_0, n_1, \dots, n_{m-1}}$. The loops couple earlier time-bin modes to later ones, and the lengths of the loops are given in time-bin units: thus a loop of length $\ell$ allows photons in modes $a$ and $a + \ell$ to interfere. This is shown below the physical schematic as a circuit. Note that by convention, the first $\ell$ time-bins are deterministically loaded into a loop of length $\ell$, and we omit drawing the initial vacuum modes inside the loop.
    \captionitem{b}~\emph{Progressive decomposition} of the circuit, where each yellow box $P_a$ contains only beamsplitters in the causal cone of the measurement of mode $a$, except for those already captured by previous causal cones. This decomposition allows us to evolve the wavefunction only through components $P_a$ that all have a bounded effective number of modes, measure their top mode, and collapse the wavefunction into a subspace.
    \captionitem{c}~An example of our lattice path formalism that describes the reachable state space of a loop-based system throughout simulation. Here, we show the state space after evolving the input wavefunction $\ket{1^5}$ through component $P_0$ and measuring 1 photon in the top mode, thus collapsing the wavefunction into a subspace shown as the yellow lattice diagram. Each basis state $\ket{n'_1, n'_2, n'_3, n'_4}$ of the space that can be found with nonzero probability corresponds to a path within the yellow diagram that starts in the bottom left corner, ends at top right corner, and takes integer steps up or right. The numbers of photons correspond to vertical steps taken by a path, as shown for two example paths on the right-hand-side. Basis states of the full Fock space $\Fs 44$ not contained in the yellow diagram are not reachable.
  }
  \label{fig:big-picture}
\end{figure*}

Boson sampling is a promising candidate in the search for quantum advantage because it presents a problem solvable by a quantum computer that generally cannot be solved classically. In single-photon Boson Sampling (SPBS) proposed by Aaronson \& Arkhipov~\cite{aaronson_arkhipov_2010}, a multi-mode Fock state made of a collection of single photons is sent through a passive interferometer and measured by photon number detectors. The probability distribution governing this process depends on the permanent of a matrix, the computation of which is \#P-hard~\cite{valiant_complexity_1979}. Under reasonable assumptions, sampling from such a distribution is classically hard~\cite{aaronson_arkhipov_2010}. As such, boson sampling is generally intractable to simulate using a classical computer beyond a few tens of photons. However, a quantum device can generate samples efficiently.

To show that a specific SPBS system achieves quantum advantage, we must determine the complexity of classically simulating the sampling process using currently known methods, and show that this is beyond the reach of current classical computers. In order to achieve sufficiently high simulation complexity, we need the interferometer to allow a large number of modes to interfere with each other, generally requiring many components to build, for example the triangular~\cite{Reck_1994} or rectangular~\cite{clements_optimal_2016} universal schemes. If the connectivity is too low, the entanglement is local and the system can be efficiently simulated, e.g. by tensor network methods~\cite{garcia-patron_simulating_2019}.

However, one family of architectures that allows such connectivity but are easier to build are loop-based systems, used for example by Madsen\etal\cite{madsen_quantum_2022} in an experiment that claimed to achieve quantum advantage. These architectures work in the time degree of freedom where modes, called time-bins, correspond to regular time intervals in which photons arrive~\cite{motes_scalable_2014,Lubasch_2018}. The interferometer consists of time-delay loops of various lengths, such as the ones illustrated in \Cref{fig:big_picture:circuit}, which allow photons in distinct time-bin modes to interfere.

Investigation of the complexity of simulating such loop-based systems is needed because, while general SPBS is known to be hard, the restricted circuit topology may be exploited to speed up classical simulation. This would raise the bar for quantum advantage in such loop-based systems. Recent examples include~Lubasch\etal\cite{Lubasch_2018}, who focused on the case where all loops have the same length, and Deshpande\etal\cite{deshpande_quantum_2021} who showed that two loops of different lengths are also not complex enough for advantage.

However, in the case of more general loop-based systems, further study is required to determine if and when these systems are simulable or not. Our work advances this field of study by establishing new conditions on the sequence of delay lines that should be used to achieve advantage. Our main contributions are described in the following.

Using a representation of a loop as a staircase-shaped sequence of possibly nonlocal gates, as illustrated in \Cref{fig:big_picture:circuit},\footnote{Note we use the same convention as~\cite{motes_scalable_2014} where we preload the loops with the first time bins which leads to the circuit as shown, omitting the initial vacuum modes of the loops.} we highlight the limited connectivity of the interferometer.
Following a light-cone argument, we decompose the circuit into smaller components shown in \Cref{fig:big_picture:progressive}. Each component has a bounded effective number of modes, and only contains a small number of beamsplitters; thus \emph{strongly} simulating a single component (by computing its state vector evolution) is easier. Our method measures a single mode after simulating each component, and collapses the wavefunction correspondingly, keeping the memory requirement tractable even for large systems. The wavefunction evolution through each component can be computed using any wavefunction simulator capable of working in the Fock (number) basis, giving a way for further specific optimizations or extensions.

To quantify the complexity of the above simulation method, and hence help in the search for quantum advantage, we develop a new formalism based on lattice paths that fully describes the state space throughout the simulation. An extension of the work of Br\'adler \& Wallner~\cite{bradler_wallner_2021}, our formalism maps one possible output state as one path in the lattice, and the entire reachable space corresponds to a region bounded by the maximal path; see an example in \Cref{fig:big_picture:lattice-paths}. Our formalism provides tools to reason not only about purely loop-based systems: one can use it to characterize the state space in an arbitrary circuit of beamsplitters and measurements, as long as the first operation is a loop of length one time-bin, and the input consists of Fock states.

The structure of our paper is the following:
First, in \Cref{sec:comp-with-other}, we compare our approach to existing methods found in the literature.
In \Cref{sec:background}, we present the background material necessary to understand the paper. Then we study the circuit topology of loop-based systems, and use it to design a new simulation algorithm in \Cref{sec:loop-structure}. In \Cref{sec:a-new-formalism} we develop our new lattice path-based formalism to characterize the state space, and in \Cref{sec:quantifying-complexity}, we show how our formalism can be used to obtain the memory complexity statistics. We use these tools to evaluate the potential for quantum advantage in a family of loop-based systems where the lengths are $1, \ell, \ell^2$. Finally, in \Cref{sec:discussion}, we contextualize our findings and identify advantage candidates that warrant further study.

We provide a companion Git repository\footnote{\label{fn:git}\url{https://github.com/orcacomputing/loop-progressive-simulator}} with a reference implementation of the simulation algorithm, as well as the lattice-path computation and complexity estimation used to obtain results in \Cref{sec:quantifying-complexity}.

\subsection{Comparison with other approaches}
\label{sec:comp-with-other}

\begin{table*}[ht]
  \centering
  \begin{tabular}{|l|c|c|c|}
    \hline
    Circuit type $\Rightarrow$ & \multicolumn{2}{c|}{\multirow{2}{*}{power-law ($\Lambda$ loops, base $\ell$)}} & full-depth Haar-random \\
    Algorithm $\Downarrow$ & \multicolumn{2}{c|}{} & triangular~\cite{Reck_1994} or square~\cite{clements_optimal_2016} \\
    \hline
    our method & $\ell = 1$ & $\mathcal{O}(m \Lambda^3 2.6^\Lambda)$ & $\mathcal{O}(m^4 2.6^{m})$ \\
    (average instance) & $\ell > 1$ & $\mathcal{O}(m \Lambda \ell^{2(\Lambda - 1)} 2.6^{\ell^{\Lambda - 1}})$ & n/a \\
    \hline
    {Clifford \& Clifford~\cite{Cliffords_2017} incl.} & \multicolumn{3}{c|}{\multirow{2}{*}{$\mathcal{O}(m 2^{m/2}) + \mathcal{O}(m^3)$}} \\
    tree decompositions~\cite{oh_classical_2022,novak_2024} & \multicolumn{3}{c|}{} \\
    \hline
    \texttt{SLOS\_full}~\cite{heurtel_strong_2023} & \multicolumn{3}{c|}{$\mathcal{O}\left( m 2.6^m \right)$} \\
    \hline
  \end{tabular}
  \caption{Computational scaling of our method compared to other simulation and sampling algorithms. Here $m$ is the number of modes, and we consider systems where the number of photons is half the number of modes. For the power-law architecture, we consider a sequence of $\Lambda$ loops of lengths $1,\ell,\ell^2,\dots,\ell^{\Lambda-1}$. The base of exponentials $2.6 \approx \sqrt{27}/2$ comes from the use of the Stirling's formula to simplify $\binom{\frac32 x - 1}{x}$ as $x$ grows, see (\ref{eq:stirling-3/2}).}
  \label{fig:comparison-table}
\end{table*}

First, we summarize how our approach compares to other relevant work in the literature. Our approach provides an efficient method for simulating loop-based circuits with a very large number of photons and modes, as long as there are few loops and their lengths are kept short. This regime is not covered by any other classical simulation algorithm to date. We describe these other algorithms in the following. Note that we focus here on exact sampling algorithms as the most relevant comparison, and do not consider methods that exploit noise to approximate boson sampling such as~\cite{oh2024classical}.

As a first point of reference, Clifford and Clifford~\cite{Cliffords_2017} developed an algorithm generating a sample from an $m$-mode $n$-photon experiment in time $\mathcal{O}(n 2^n) + \mathcal{O}(mn^2)$. This general method does not exploit any internal structure of the interferometer and only depends on the number of modes and photons. As such, it becomes impractical in regimes above 40 photons, even when the interferometer is highly structured. Our algorithm can sample from states with many more photons in simple loop-based systems.

For settings where local-only interactions induce a sparse graph structure, tree decomposition techniques can reduce the complexity of the Clifford and Clifford algorithm to $\mathcal{O}(m n^2 \omega^2 2^\omega)$~\cite{oh_classical_2022}, with a further reduction to $\mathcal{O}(n^2 2^\omega \omega^2) + \mathcal{O}(\omega n^3)$ achievable in the 1D lattice case~\cite{novak_2024}. Here $\omega$ corresponds to the (highest) treewidth of the graphs representing measurement outcomes. However, in general, these tree decomposition techniques are not transferrable to our setting: in the presence of loops of arbitrary lengths, long-range interference occurs which leads to less sparsity in the graph representing the connectivity of modes. As such, they cannot be used for the loop-based circuits considered in this work.

Another point of reference is the \texttt{SLOS\_full}~\cite{heurtel_strong_2023} \emph{strong} state vector simulator. The \texttt{SLOS\_full} algorithm achieves a running time $\mathcal{O}(n \binom{n + m - 1}{m - 1})$ where the binomial is the dimension of the Fock space. It is general in the sense that it does not exploit any internal structure. This, however, means that at some number of photons and modes, the time and memory requirements become intractable. Our algorithm, on the other hand, calls a strong simulation subroutine only on a smaller subcircuit, with a smaller number of photons, and then samples some of the photons, thus keeping the problem tractable even for many photons and modes in a loop-based system. Note that on a Haar-random circuit, our method is slower than \texttt{SLOS\_full} since we evolve the state vector over all the beam splitters one by one instead of just using the matrix description of the overall unitary to calculate the state.

In \Cref{fig:comparison-table}, we summarize the comparison in a table showing the time complexities of our algorithm and the other algorithms described above. As illustrative examples we focus on two cases. First, we consider the \emph{power-law} architecture with the sequence of loops of lengths $1, \ell, \ell^2, ... , \ell^{\Lambda-1}$ (with integer $\ell$ and $\Lambda$) that we study in \Cref{sec:comparison-power-law-archs}, with an input state $\ket{\vec n} = \ket{1,0,1,0,\dots}$. Ours is the only method that does not scale exponentially with the number of photons or modes, with scaling instead with $\ell$ and~$\Lambda$. We also contrast this with a Haar-random circuit, modelled as a triangular or rectangular schemes~\cite{Reck_1994,clements_optimal_2016} which can also be implemented in a loop-based system (setting $\ell=1$), and the same input state. For these Haar-random circuits, other methods are more appropriate. The derivation of the complexities presented in this table can be found in \Cref{sec:theor-compl-progr}.

We note that the theoretical hardness of loop-based circuits has been explored by Go\etal~\cite{go_exploring-shallow-depth_2024,go_computational_2024}. In~\cite{go_exploring-shallow-depth_2024}, the authors present a circuit with a ``non-local hypercubic structure'' which can be implemented with a power-law architecture with sequential loops of length $1,2,4, ... , 2^{\left\lfloor \log_2(m) \right\rfloor}$ where $m$ is the number of modes. This circuit is shown to exhibit properties that are conducive to hardness, which include the ability to mimic a Haar-random unitary up to the first moment and the so-called \emph{hiding} property exploited in the usual hardness proofs for boson sampling. In~\cite{go_computational_2024}, a circuit with two repetitions in sequence of this power-law architecture was shown to be hard. Our work bridges some of the gap between known classical simulation methods for one loop \cite{Lubasch_2018} and two loops \cite{deshpande_quantum_2021} and the hardness results from Go\etal by extending the range of simulable loop-based architectures.

\section{Background}
\label{sec:background}

We start by introducing some background to further motivate the study of loop-based time-bin interferometers and to introduce our new methods to simulate and characterize them. We give a short overview of SPBS in \Cref{sec:short-intro-bs}, we motivate and define loop-based time-bin interferometers in \Cref{sec:loops}, and we present the background for the lattice path formalism in \Cref{sec:LP-background}.

\subsection{Single-photon Boson Sampling}
\label{sec:short-intro-bs}

Single-photon Boson Sampling (SPBS) is a quantum optical experiment defined by Aaronson \& Arkhipov in~\cite{aaronson_arkhipov_2010} where we send a collection of $n$ indistinguishable single photons into a linear interferometer on $m$ modes (ports) and measure where they exit.
The quantum state of light in such an experiment is described by a vector in \emph{Fock space} $\Fs nm$. This is spanned by the number basis made of vectors $\ket{n_0, \dots, n_{m-1}}$ where each $n_a$ is the number of photons in mode $a$. The dimension of a Fock space of $n$ indistinguishable photons distributed in $m$ modes is:
\begin{equation}
  \label{eq:dim-of-Fock-space}
  \dim \Fs nm = \binom{m - 1 + n}{n}
\end{equation}

The SPBS problem is defined as follows: generate Fock basis states measured after sending an input $\fket{\vec n}$, usually with $n_a \le 1$ for each mode, into a linear interferometer. The interferometer is a mode-to-mode transformation described by a unitary matrix $U : \C^m \to \C^m$. There are many ways to construct these, most commonly decomposing them as a system of single-mode phase-shifters and two-mode beamsplitters~\cite{Reck_1994,clements_optimal_2016}. We represent a beamsplitter $B(\theta) : \C^2 \to \C^2$ as the following rotation matrix~\cite{Campos_1989}:
\begin{equation}
  \label{eq:BS-on-modes}
  B(\theta) =
  \begin{pmatrix}
    \cos \theta & \sin \theta \\
    -\sin \theta & \cos \theta
  \end{pmatrix}.
\end{equation}
In a system of more than two modes, we specify that a beamsplitter acts on modes $a,b$ as $B_{a,b}(\theta)$. Note that we omit any reflection or transmission phases, only representing the coupling angle $\theta$. Furthermore, we will often care only about a coupling being possible but not the concrete value of~$\theta$, in which case we omit it from the notation.

The input-output transition amplitudes of the interferometer $U$ acting on the Fock space are permanents of certain matrices derived from~$U$~\cite{aaronson_arkhipov_2010}.
The computation of a permanent of a general matrix belongs to the complexity class \#P-hard~\cite{valiant_complexity_1979,aaronson_arkhipov_2010}. This means that classically simulating Boson Sampling is likewise hard, and that is the reason why Boson Sampling devices, which can generate these samples natively, are candidates for quantum advantage.

\subsection{Loop-based time-bin interferometers}
\label{sec:loops}

One specific type of architecture used for Boson Sampling, proposed by Motes\etal~\cite{motes_scalable_2014} and experimentally implemented for example by Madsen\etal~\cite{madsen_quantum_2022} and Carosini\etal~\cite{carosini_2024}, is based on loops. Here, some output port of a device is connected back to an input port of the same device using an optical loop; see an example in \Cref{fig:big_picture:circuit}.
The modes of a loop-based interferometer, called \emph{time-bins}, correspond to regular time intervals when photons arrive or leave~\cite{Lubasch_2018}. This is useful because it allows us to build complex interferometers without using many physical components~\cite{PhysRevA.109.042613}, and the combination of being practically realizable, but at the same time allowing large interferometers, makes loop-based architectures candidates for quantum advantage demonstrations.

\subsubsection{Single loop}
\label{sec:single-loop}

We define an optical loop as a device consisting of a delay line and a single beamsplitter which may change its transmissivity in time. The loop connects one output port of the beamsplitter to an input; the remaining two ports are the input and output of the device.

We understand this as a circuit consisting of a staircase-shaped sequence of beamsplitters, as in \Cref{fig:big_picture:circuit}. The length of the loop is given in time-bin units, so if the loop has length $\ell$, denoted $L_\ell$, then it couples mode $a$ to $a + \ell$. Note that for every loop $L_\ell$, including when these are composed in \cref{sec:sequential-composites-loops}, we employ the convention of Motes\etal\cite{motes_scalable_2014} where photons from the initial $\ell$ modes deterministically enter the loop, and we omit the latter's original vacuum modes. This gives us the circuit in \cref{fig:big_picture:circuit}, and the following definition:

\begin{definition}[optical loop]
  Let $\ell \ge 1$. Then define a loop of length $\ell$ as the following composite of time-bin beamsplitters:
  \[ L_\ell = B_{0,\ell} \then B_{1, 1+\ell} \then \cdots \then B_{a, a+\ell} \then \cdots, \]
  where the composition is written from left to right ($A \then B \deq B A$ for maps $A,B$), and we omit the coupling angles. In practice, the total number of modes $m$ is finite and the above composite stops at $B_{m-1-\ell, m-1}$.
\end{definition}

The simplest example is the loop $L_1 = B_{0,1} \then B_{1,2} \then \cdots$, which looks like a staircase of nearest-neighbour beamsplitters when drawn as a circuit (see the left part of~\Cref{fig:big_picture:circuit}). This loop is important throughout the paper, because it establishes a shared pool of photons to be distributed between modes; we come back to this idea in \Cref{sec:LP-background,sec:a-new-formalism}.

\subsubsection{Sequential composites of loops}
\label{sec:sequential-composites-loops}

A single photonic loop $L_\ell$ is a fairly simple system, and in particular it can be efficiently simulated~\cite{Lubasch_2018}. This means that in order to achieve quantum advantage, we need to look elsewhere, and the natural idea for a loop-based setup is to compose multiple loops together.

There are different ways that multiple loops could be joined together, and in our work we focus on \emph{sequential} composites of loops $L_{\ell_1} \then L_{\ell_2} \then \dots \then L_{\ell_\Lambda}$, where $\Lambda$ is the number of loops, and the lengths $\ell_j$ may be different. An example of such architecture is $L_1 \then L_3$ shown in~\Cref{fig:big_picture:circuit}.
Note that if all lengths $\ell_j = 1$, then the individual loops are similar to slices of the Reck scheme~\cite{Reck_1994}, a well-known universal decompositions of optical unitaries to a system of beamsplitters and phase shifters.

\begin{definition}[power-law composites]
  \label{def:power-law}
  An interesting family of sequential composites of loops are those with lengths $\vec\ell = (1, \ell, \ell^2, \dots, \ell^{\Lambda - 1})$, which we call \emph{power-law} composites. We call $\ell$ the \emph{base loop length}.
\end{definition}

Power-law systems are important because of their simple form that still allows long-range entanglement in shallow circuits ($\Lambda$ small). The property that each length is an exponent of $\ell$ generates connectivity between modes that can be understood as the volume of a $\Lambda$-dimensional hypercube with sides of length $\ell$, where modes are vertices at integer coordinates and edges correspond to beamsplitters~\cite{deshpande_quantum_2021,go_exploring-shallow-depth_2024}.
Such architectures have been recently used in an advantage demonstration~\cite{madsen_quantum_2022}, and in Boson Sampling hardness results~\cite{go_computational_2024,go_exploring-shallow-depth_2024}.

\subsection{Lattice path formalism for bosonic Fock spaces}
\label{sec:LP-background}

The circuit topology of loop-based systems restricts where photons can go, and this can be captured by a formalism where the state space corresponds to a region in an integer lattice. Possible output number states correspond to lattice paths within this region (see for example \Cref{fig:big_picture:lattice-paths}), and thus the lattice-path representation can be used to count the state space using combinatorial methods.

In this section, we describe the original lattice-path formalism developed by Br\'adler \& Wallner~\cite{bradler_wallner_2021} that can be used to count and enumerate the state space of a system of loops that all have length one. We expose their results using a different notation and terminology that will suit our needs later; and we then expand their formalism to general loop-based systems in \Cref{sec:a-new-formalism}.

\subsubsection{Lattice paths}
\label{sec:lattice-paths}

Informally, a \emph{lattice} is a collection of discrete points in space that repeat in regular intervals. We will be interested in lattices contained within finite rectangles in the $\N^2$ quadrant of $\Z^2$, and in the paths allowed within these rectangles.

\begin{definition}[lattice path]
  \label{def:lattice-paths}
  Let $W, H \in \N$. A~\emph{lattice path} in the rectangle of dimensions $(W,H)$ is a path in $\N^2$ that starts at the origin $(0,0)$, ends at point $(W,H)$, and only takes steps right $(+1, 0)$ or up $(0, +1)$.
  The set of all $(W,H)$-lattice paths is denoted $\LL(W,H)$.
\end{definition}

\begin{figure}[t]
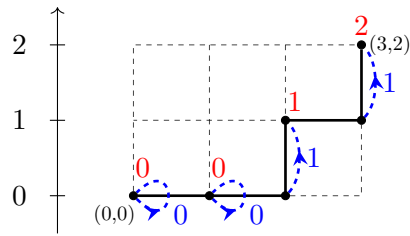

    \ctikzfig{LP_RRURU_lambda-and-n}
  \caption{Example of a lattice path from $\LL(3,2)$ in both the height and step representations.
    \captionitem{\red{red}}~Height representation $\lambda = (0,0,1,2)$ where the coordinates indicate the height of the path above the horizontal axis. This is the maximum $y$ coordinate of a point passed by the path at each $x$ coordinate.
    \captionitem{\blue{blue}}~Step representation $\vec n^\lambda = \Delta\lambda = (0,0,1,1)$ with coordinates $\vec n^\lambda_a = \lambda_a - \lambda_{a-1}$ that represent the size of the vertical steps the path takes, shown by blue dashed arrows.
  }
  \label{fig:LP-example}
\end{figure}

To describe individual paths in $\LL(W,H)$, we define two equivalent representations, each with a different domain of usefulness. We show an example path and both representations in \Cref{fig:LP-example}.

\begin{definition}[height representation]
  In \emph{height representation}, a path is written as a nondecreasing sequence of integers $\lambda = (\lambda_0, \dots, \lambda_{W})$. Each component $0 \le \lambda_a \le H$ is defined as
  \[
    \lambda_a = \max\{ y_i : (a, y_i) \in \N^2 \ \text{is passed by}\ \lambda \},
  \]
  where $(a, y)$ are points with $x=a$ passed by the path; see example in \Cref{fig:LP-example} in \red{red}. Note that since a path must end at point $(W,H)$, the final component is always $\lambda_W = H$.
  Height representation is denoted by lowercase Greek letters $\lambda,\kappa,\mu,\dots$
\end{definition}

\begin{definition}[step representation]
  \label{def:lp-step-rep}
  A lattice path in $\LL(W,H)$ represented by $\lambda$ in height representation has \emph{step representation} $\vec n^\lambda = (n^\lambda_0, \dots, n^\lambda_W)$ where the component $n^\lambda_a$ is the number of steps up that $\lambda$ takes at $x=a$. Equivalently, it is the discrete difference of the height representation, written $\vec n^\lambda = \Delta \lambda$ with components
  \[ n^\lambda_a = \lambda_a - \lambda_{a-1}, \]
  see example in \Cref{fig:LP-example} in \blue{blue}. By convention, $\lambda_{-1} = 0$.
  The step representation is suggestively denoted by $\vec n$ to indicate the connection to bosonic Fock states, and identified by a superscript whenever required.
\end{definition}

Both representations are clearly equivalent, and to go from step representation $\vec n$ to height representation $\lambda^{\vec n}$ (note the superscript to identify the path), we take the summation:
\begin{equation}
  \label{eq:n-to-lambda}
  (\lambda^{\vec n})_a = \sum_{b=0}^a n_b.
\end{equation}

\subsubsection{Fock states represented by lattice paths}
\label{sec:fock-stat-repr}

In \Cref{def:lp-step-rep}, we define the step representation of a lattice path, suggestively denoting it $\vec n$ highlighting the connection to Fock (number) states. We show two example basis states in the right path of \Cref{fig:big_picture:lattice-paths}. We now make this connection formal:

\begin{lemma}[bosonic states are lattice paths]
  \label{lem:Fock-space-gen-by-LP}
  The bosonic Fock space $\Fs nm$ of $n$ indistinguishable photons in $m$ modes is spanned by lattice paths in a rectangle of size $(m-1, n)$:
  \begin{equation}
    \label{eq:Fock-space-lattice-paths}
    \Fs nm = \gen{ \LL(m-1, n) }_\C,
  \end{equation}
  where we use the notation $\gen{S}_\C$ (or just $\gen S$) for a $\C$-vector space spanned by elements of a set $S$ seen as an orthogonal basis.
\end{lemma}
\begin{proof}
  We show that the dimension is correct. A lattice path in $\LL(W,H)$ must start at the origin and take $H$ steps up and $W$ steps right to reach $(W,H)$. It takes a total of $W+H$ steps, and any order of the steps is allowed, so the number of paths is the combination number $\binom{W+H}H$. Setting $W=m-1$ and $H=n$:
  \[
    |\LL(m-1, n)| = \binom{m-1 + n}{n}
    \overset{\eqref{eq:dim-of-Fock-space}}{=} \dim \Fs nm.
  \]
  Finally, note that all lattice paths end at height~$n$, so each path in step representation, which contains the sizes of the steps up, corresponds to distributing the $n$ photons into $m$ modes.
  We conclude that lattice paths in $\LL(m-1,n)$ indeed correspond to basis states of $\Fs nm$.
\end{proof}

So far, we have a convenient way to visually represent number states that implements the constraint that there are in total $n$ photons distributed between $m$ modes. In the next section, we show that lattice paths can help us implement more constraints on the state space.

\subsubsection{Loops give downsets}
\label{sec:loops-give-downsets}

The circuit topology of a loop-based system allows incoming photons to end up only in some modes. We can capture this using a region in the lattice that represents the subspace of the valid number states. As shown in \Cref{fig:example-max-photons-to-LP}, the important new information encoded by this region are the constraints on the maximal number of photons in a given mode, which form a nondecreasing sequence that defines a lattice path, and all paths below it are allowed. This is enough to describe the state space because photons are indistinguishable and we only know how many are in each mode.

We show how this works for a system with length \mbox{$\ell_j = 1$} for each loop, which is one of the main results of~\cite{bradler_wallner_2021}. Our work in \Cref{sec:a-new-formalism} extends this to more general systems.

We take the simplest example using a system of a single loop $X = L_1$ acting on an input state $\fket{1^5}$ on 5 modes in \Cref{fig:example-max-photons-to-LP}.
At each of the output wires of the circuit, we write down the set of possible numbers of photons measured there: these have the form $\iicc{0, \mu_a} = \{0, 1, \dots, \mu_a\}$, where $\mu_a$ denotes the maximum number of photons that could exit in mode~$a$. Intuitively, $\mu_a$ can be found by following the possible paths of the photons through the circuit $X$. More formally, if the input state is $\fket{\vec n}$, we can express the maximum number of photons in output mode~$a$ as
\[
  \mu_a = \sum_b
  n_b,
\]
where the summation runs over the input modes~$b$ from which photons can reach output mode $a$.

\begin{figure}[htp]
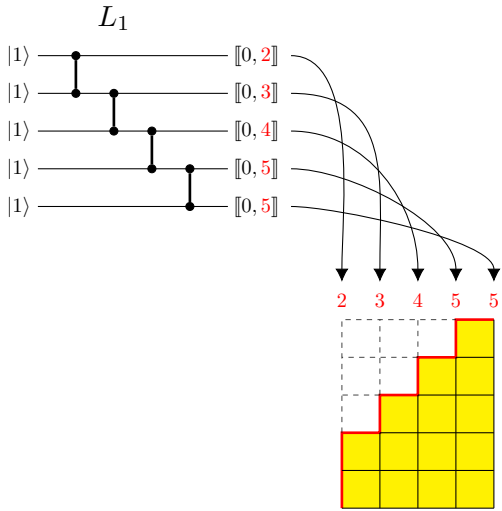

  \centering
  \ctikzfig{L1_1s_maxLP}
  \caption{An example showing how we obtain a maximal lattice path $\mu$ from a given circuit and input state, in this case $L_1$ and $\fket{1^5}$. On the output modes are sets of possible photon numbers measured there, and we show how these define the lattice path $\mu$ (red). The possible measurement outcomes are all the lattice paths lower than $\mu$ -- these are the paths constrained to the yellow-shaded region under $\mu$.
  }
  \label{fig:example-max-photons-to-LP}
\end{figure}

Due to the staircase structure of $L_1$, the sequence $\mu = (\mu_0, \dots, \mu_{m-1})$ is nondecreasing and terminates in $n$, the total number of photons. This means it corresponds to a height representation of a lattice path in $\LL(m-1, n)$, shown in \red{red} in the example of \Cref{fig:example-max-photons-to-LP}. The set of possible output basis states is exactly the set of paths lower than $\mu$, i.e. those constrained to the yellow-shaded region in the figure. We refer to $\mu$ as the \emph{maximal lattice path}, and we call the set of paths below or equal to it the \emph{downset} $\downset\mu$. This exactly describes the state space which has the following form:
\begin{definition}[cumulative space]
  \label{def:cumulative-space}
  Let $\mathcal{K}$ be a subspace of $\Fs nm$. If there exists a maximal lattice path $\mu \in \LL(m-1, n)$, corresponding to a basis state $\ket{\vec n^\mu}$ in $\Fs nm$, such that $\mathcal{K} = \gen{\downset\mu}_\C$, then we call $\mathcal{K}$ a \emph{cumulative space}.
\end{definition}

\subsubsection{Interpretation}

The reason to make the above definition and call it a cumulative space is that the loop $L_1$ acts as an \emph{accumulator} of photons. It collects incoming photons, and may release them into any subsequent time-bin. Notice that the maximal path $\mu$ corresponds to photons leaving the system as soon as possible, i.e. never entering the loop.\footnote{Except of course the first photon, which is preloaded into the loop by convention (see \cref{sec:single-loop}). This photon also immediately exits.} The paths below $\mu$ correspond to some photons entering the loop and leaving later than their starting time-bin, with the extreme being the bottom path $\lambda^\bot = (0,\dots,0,n)$ where all photons stay in the loop until the last time-bin mode.

Finally, we shortly note that in this context of lattice paths, \emph{lower than} has a precise meaning in terms of a partial order on the set $\LL(m-1, n)$, called the \emph{Young order}. Two paths $\lambda$ and $\mu$ satisfy $\lambda \le \mu$  iff $\lambda_a \le \mu_a$ for each component $a$. The heights of $\lambda$ are lower in the lattice than those of $\mu$, so indeed $\lambda$ belongs to the region below $\mu$ in \Cref{fig:example-max-photons-to-LP}, which is the downset
\[ \downset\mu = \{ \lambda \in \LL(m-1, n) : \lambda \le \mu \}. \]
In the main body, we give a mostly non-order-theoretic presentation, and instead we only elaborate on the formal details in \Cref{app:lattice-path-formalism-details}.

\section{Loop-based structure allows simulation of certain interferometers}
\label{sec:loop-structure}

We present a new wavefunction evolution algorithm to classically simulate loop-based interferometers exploiting their staircase structure by performing partial measurements when modes become unused.

\subsection{Progressive simulation algorithm}
\label{sec:simulation-algorithm}

Observe in \Cref{fig:big_picture:circuit} that due to the staircase shape of the circuits corresponding to loops, each mode has a time in the circuit after which it is never affected by any subsequent operation. Our algorithm is a hybrid between strong and weak simulation: it evolves the wavefunction through small components of the circuit, and after each of them, it simulates the measurement of a single mode that is never needed again, collapsing the wavefunction to a smaller subspace in the process. Due to the collapse, the memory required to store the state vector stays more manageable. Furthermore, we show that the strategy of measuring single modes and collapsing the state means that we only strongly simulate circuits with a low effective number of modes. Therefore our algorithm can tackle even large loop-based systems that could not be simulated otherwise.

\subsubsection{Decomposition of the circuit}
\label{sec:progressive-decomposition}

\begin{figure*}[ht]
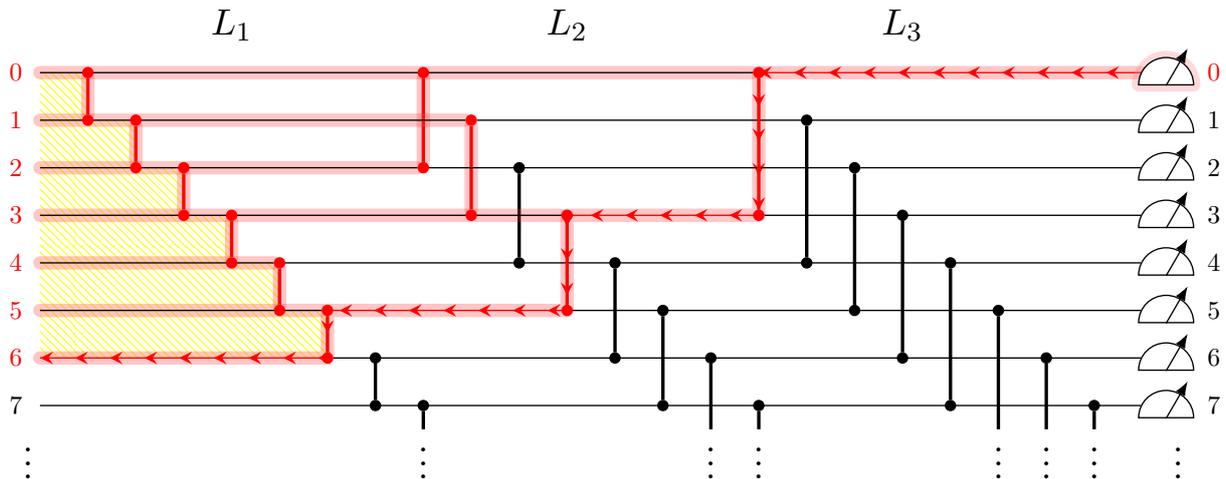

  \fwtikzfig{L1-2-3_preimage-of-0}
  \caption{A view of the input modes and beamsplitters influencing output mode 0 in a system $L_1 \then L_2 \then L_3$. We highlight in red the wires along which photons may reach output mode 0 on the right. The path reaching the maximum input mode is further emphasized with arrows ($\tikzredarrowsback$). The relevant beamsplitters are drawn in red as well -- these are the beamsplitters that form the component $P_0$. Observe that the first loop $L_1$ is a staircase of nearest-neighbour beamsplitters, so if a mode $a$ is relevant, so are the modes $<a$; this fact is emphasized in the yellow area.}
  \label{fig:preimage-of-0}
\end{figure*}

We describe the idea using the simple system $L_1 \then L_3$ shown in~\Cref{fig:big_picture:circuit}.
We push the measurement of mode 0, and all beamsplitters that influence it, to the left. At the same time, we push all subsequent beam splitters to the right of the measurement, obtaining \Cref{fig:big_picture:progressive}. Two beamsplitters can be pushed past each other whenever they do not act on the same modes. Note that moving measurements to the left, we are identifying their causal cones, as shown in the more complex example of \Cref{fig:preimage-of-0} where we highlight the paths in the circuit through which photons could arrive to output mode 0.

In \Cref{fig:preimage-of-0}, notice that to find all the modes and beamsplitters relevant to output $0$, we just need to find the farthest reachable mode.  We follow the path with arrows ($\tikzredarrowsback$) from output $0$ backward, and use each beamsplitter to jump to a later time-bin mode until we reach the first loop of length 1. Since this loop corresponds to a staircase of nearest-neighbour beamsplitters, all the time-bin modes identified by lower numbers are also connected to output $0$ (shown by yellow fill).

We iterate the procedure of moving measurements left for all output modes in the order $0,1,\dots$, transforming the example circuit in \Cref{fig:big_picture:circuit} to the form of \Cref{fig:big_picture:progressive} which we call the \emph{progressive decomposition} to components $(P_a)_{a=0}^{m-1}$. Each component $P_a$ contains only the beamsplitters inside the causal cone of mode~$a$ that have not appeared in the previous components $P_0, \dots, P_{a-1}$. This means that even though the set of beamsplitters relevant to an output mode $a$ includes some of the beamsplitters relevant for mode $a-1$, the component $P_{a}$ contains only the ``new'' beamsplitters.

After rewriting the circuit as above, observe the effective number of modes in all the components~$P_a$ is bounded by a number which only depends on the loop lengths. In \Cref{fig:big_picture:progressive}, this is the maximum height of a box $P_a$.

\begin{theorem}[relevant modes]
  \label{thm:relevant-modes}
In a progressive decomposition of a loop-based system $X$ made of loops of lengths $\vec\ell = (\ell_1, \dots, \ell_\Lambda)$, each component $P_a$ acts non-trivially on at most the number of \emph{relevant} modes
  \begin{equation}
    \label{eq:relevant-modes}
    R = 1 + \sum_{i = 1}^\Lambda \ell_i.
  \end{equation}
\end{theorem}

This follows from the intuition above; we also provide a formal proof in \Cref{app:loop_theorems}.

\subsubsection{Simulation}

\begin{myalgorithm}{%
    Pseudocode of the progressive simulation algorithm. This algorithm is independent of how the state is stored (e.g. as a dictionary).
  }{\label{fig:pseudocode-prog-sim}}
  \SetAlgoRefName{PS}
  \KwIn{Number of modes $m$, vector of loop lengths $\vec \ell = (\ell_1, \dots, \ell_\Lambda)$, vector of beamsplitter angles $(\theta_i)_{i}$, input state $\vec n = (n_0, \dots, n_{m-1})$.}
  \KwOut{A single sample $\vec n' = (n'_0, \dots, n'_{m-1})$.}

  $(P_a)_{a=0}^{m-1} \gets$ progressive decomposition of $L_1 \then \cdots \then L_\Lambda$ obtained by pushing measurements to the left and reordering beamsplitters (together with their angles~$(\theta_i)_i$) accordingly \;

  $\ket{\psi} \gets \ket{\vec n}$\;

  \For{$a = 0, \dots, m-1$}{
    $\ket\psi \gets$ evolve $\ket{\psi}$ through component~$P_a$ as a subroutine \; \label{alg:progressive-simulation:ln:evolution}

    $p(n'_a) \gets$ compute the marginal probability of measuring $n'_a$ photons in output mode $a$, for each possible $n'_a$\;

    $n'_a \gets$ random sample from the distribution $p$\;

    $\ket\psi \gets$ collapse $\ket\psi$ to subspace compatible with measurement $n'_a$\;
  }

  \KwRet{$\vec n' = (n'_0, \dots, n'_{m-1})$}\;

  \caption{Progressive simulation}
  \label{alg:progressive-simulation}
\end{myalgorithm}

We can understand the progressive decomposition as a sequence of smaller circuits $P_a$, each of which only acts on at most $R$ modes, and only includes a few beamsplitters. We evolve the state vector through only these beamsplitters, and then we perform a partial measurement of the top mode $a$ of the component, collapsing the wavefunction to a subspace compatible with the measurement outcome. The collapsed wavefunction becomes the input state of the next component.
We iterate the procedure for all components, thus progressively generating the full sample $\ket{\vec n'}$.
We summarize the algorithm, labelled Algorithm~PS, as pseudocode in \Cref{fig:pseudocode-prog-sim}.

Observe in \cref{alg:progressive-simulation:ln:evolution} in the pseudocode that we call the strong simulation as a subroutine. This gives us the freedom to choose any wavefunction simulation algorithm, as long as it can accept a superposition state as input. This is required because the projected state after measuring a mode~$a$, which is the input to $P_{a+1}$, is generally a superposition of several number states.

\subsection{Evaluating the algorithm}
\label{sec:evaluating-algorithm}

To evaluate how useful the progressive algorithm is, we aim to estimate the cost of using it as a function of different parameters of the system to be simulated.

We specifically focus on the memory complexity because this gives us a lower bound on the running time for any choice of exact wavefunction simulator. This is the case because any wavefunction simulator must at least spend time to read the entire stored state. Focusing on memory is further justified by the fact that the current best algorithm \texttt{SLOS\_full} performs wavefunction simulation in time linearly bounded by memory~\cite{heurtel_strong_2023}.

The three main contributions to memory complexity come from the following:
\begin{enumerate}
\item The effective number of modes in each subcircuit $P_a$, given by the value $R$ from \Cref{thm:relevant-modes}.
\item The total number of photons in the system.
\item The possible positions of photons, given by the input state and the circuit topology.
\end{enumerate}

Point number 3 is the focus of the next \Cref{sec:a-new-formalism} and one of the main technical contributions of the present paper. However, already using the first two points, we can roughly estimate which architectures are easy to simulate. For example, we can clearly hold a billion coordinates in the memory of a high-end computer, so we set the dimension to  $10^9$. Assume, as an upper bound, that we need the full Fock space with dimension in eq.~\eqref{eq:dim-of-Fock-space}. We need to simulate $R$ modes, and for simplicity assume we have $R$ photons to maximize the binomial in \eqref{eq:dim-of-Fock-space}. Then R should satisfy
\[
  \binom{2 R - 1}{R} \approx 10^9,
\]
which is the case when $R \approx 16$. This implies we can easily simulate systems made of loops of lengths, for example, $\vec\ell = (1, 14)$, $(1, 5, 9)$, or $(1,2,4,8)$.

However, this method is a very coarse approximation. Not only may there be fewer or more photons in a component $P_a$ to be simulated, invalidating the assumption that $n \approx R$ above, but the circuit topology restricts where those photons can be. This means the actual state space is generally very different.
We present a more granular and robust way to estimate complexity in the following sections.

\section{A new formalism to enumerate the state space}
\label{sec:a-new-formalism}

To quantify the complexity of our progressive simulation algorithm, in \Cref{sec:evol-thro-long} we expand the formalism of Br\'adler \& Wallner~\cite{bradler_wallner_2021} from loops of length one, presented in \Cref{sec:LP-background}, to more general systems where loops have different lengths. Then in \Cref{sec:meas-arbitr-modes}, we show how to encode the outcomes of mid-circuit partial measurements. Together, these allow us to describe the state space throughout simulation, an important tool in computing the complexity in \Cref{sec:quantifying-complexity}.

\subsection{Evolution through long loops}
\label{sec:evol-thro-long}

\begin{figure}[tb]
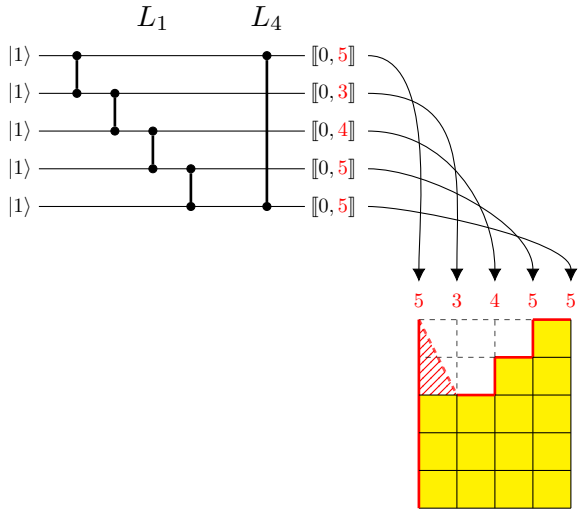

  \centering
  \ctikzfig{L1-4_defect}
  \caption{Example circuit that cannot be described by the original lattice path formalism of \protect\cite{bradler_wallner_2021}. The maximum numbers of photons form a sequence $\mu = (5,3,4,5,5)$ which is not nondecreasing, and hence not a lattice path in height representation. As a lattice path, $\mu$ would have to descend (hatched part in the picture) and this is not allowed.}
  \label{fig:lattice-path-defect}
\end{figure}

Recall from \Cref{sec:loops-give-downsets} that we describe the possible outputs of a loop-based system by a sequence $\mu = (\mu_0, \dots, \mu_{m-1})$ where each $\mu_a$ is the maximum number of photons that can be measured in output mode $a$. For a system of loops of length one, Br\'adler and Wallner~\cite{bradler_wallner_2021} show that $\mu$ defines a lattice path in height representation, and that the state space is fully characterized by the set of lattice paths lower than $\mu$. In order for $\mu$ to be a valid lattice path, it must be a nondecreasing sequence.

This breaks when we allow a loop of longer length. See for example in \Cref{fig:lattice-path-defect} that the loop $L_4$ allows the exchange of photons between the top and bottom mode, leading to a maximal ``path'' $\mu = (5,3,4,5,5)$. As $\mu$ is not nondecreasing, it is not in fact a valid lattice path. We show that this issue can be addressed by introducing mode permutations into this formalism.

\subsubsection{Lattice paths with mode permutations}

\begin{figure*}[tb]
  \centering
  \ctikzfig{non-nn-bs}
  \caption{An example of the PCS evolution rule from \Cref{thm:PCS-BS-evolution}, following the same steps as the accompanying text. We start with a PCS $\gen{\downset\mu}^\sigma$ obtained by the circuit and input from \Cref{fig:lattice-path-defect}. On the left, we display the maximum path $\mu$ and the permutation $\sigma$ acting on the PCS. Next, we apply the permutation $\sigma$ to the top path to obtain the vector $v$ describing the maximum photon numbers in real (unpermuted) modes, and we highlight that this is not a lattice path. Then we apply a beamsplitter between modes $0$ and $2$, obtaining the new vector of maximum photon numbers $w$, which is again not a lattice path. Finally, we sort the modes using a permutation $\tau^{-1}$ to obtain a new maximum lattice path $\kappa$, landing in a new PCS $\gen{\downset\kappa}^\tau$. Note that by construction the permutation acts to permute modes from PCS to the physical space, and hence sorting the modes from physical to PCS in the last step is done by $\tau^{-1}$.}
  \label{fig:non-nn-bs}
\end{figure*}

To encode a loop of length $\ell_j > 1$, we introduce mode permutations in the lattice path representation so that they appear in an order of nondecreasing maximum photon number. Recall that in \Cref{sec:loops-give-downsets}, the nondecreasing property of the maximal path $\mu$ arises due to a loop $L_1$ accumulating incoming photons, and releasing them in any subsequent time-bin modes. Individual lattice paths in $\downset\mu$, via their step representation, capture ways to distribute the fixed number of photons between modes. However, the order of modes in the lattice path representation has no physical meaning, serving only as labels that can be reordered.

In the following, we show that even with modes permuted, we can still fully describe the state space by a downset $\downset\mu$. To do this, we generalize the notion of cumulative space from \Cref{def:cumulative-space}:

\begin{definition}[PCS]
  \label{def:PCS}
  A subspace $\mathcal{H} \subseteq \Fs nm$ is called a \emph{permuted cumulative space (PCS)} if there exists a cumulative space $\gen{\downset\mu} \subseteq \Fs nm$ and a permutation of modes $\sigma \in S_m$ that acts as an isomorphism from $\gen{\downset\mu}$ to $\mathcal{H}$ sending
  \[
    \ket{n_0, \dots, n_{m-1}} \mapsto \neket{n_{\sigma(0)}, \dots, n_{\sigma(m-1)}}.
  \]
  We write $\mathcal{H} = \gen{\downset\mu}^\sigma$.
\end{definition}

\subsubsection{Non-nearest-neighbor beamsplitter}

To show how to use permuted spaces to reason about general loop-based systems, we focus first on describing the action of a basic building block, a single beamsplitter between arbitrary modes $a$ and $b$ that are not necessarily nearest neighbors. In \Cref{fig:non-nn-bs}, we show an example of the same steps that follow in the explanation.

We define a vector $v$ where $v_c = \mu_{\sigma(c)}$~is the maximum possible number of photons in mode $c$. Due to the cumulative nature of the input state space, the maximum number of photons in either modes $a$ and $b$ after applying the beamsplitter is $\max\{v_a, v_b\}$. We write a new vector of maximum photon numbers $w_c$ with components:
\[
  w_c =
  \begin{cases}
    \max\{ v_a, v_b \} & \text{if}\ c \in \{a,b\}, \\
    v_c & \text{if}\ c \notin \{a, b\}.
  \end{cases}
\]

Since $a$ and $b$ are not necessarily nearest neighbors, and because $v$ may not be nondecreasing in the first place, it is possible that $w$ is not a nondecreasing sequence. Hence $w$ is not a valid lattice path. We sort it in nondecreasing order and thus permute the modes. We call the permuted sequence $\kappa$ which is a valid lattice path in height representation, and we call $\tau$ the inverse of the permutation we used (where the permutation maps $\kappa \mapsto w$).

We summarize the procedure as the following evolution rule:
\begin{theorem}[PCS evolution]
  \label{thm:PCS-BS-evolution}
  Suppose $\mathcal{H} = \gen{\downset\mu}^\sigma$ is a PCS with a mode permutation $\sigma$ applied. Then applying a beamsplitter between arbitrary modes $a,b$ maps $\mathcal{H}$ to a PCS $\gen{\downset{\kappa}}^\tau$ with a new mode permutation $\tau$, where
  \begin{equation}
    \label{eq:PCS-BS-evolution}
    \kappa_{\tau(c)} =
    \begin{cases}
      \max\{ \mu_{\sigma(a)}, \mu_{\sigma(b)} \} & \text{if}\ c \in \{a, b\}, \\
      \mu_{\sigma(c)} & \text{if}\ c \notin \{a, b\};
    \end{cases}
  \end{equation}
  and where $\tau$ assures that $\kappa$ is non-decreasing.
\end{theorem}
\begin{proof}
  This follows from the reasoning above the theorem statement.
\end{proof}

It follows that to obtain the state space at the output of a loop-based system, we apply the above rule iteratively. This simplifies to the following:

\begin{corollary}
  \label{cor:final-lattice-path}
  Suppose $X$ is a loop-based system made of loops of lengths $\vec\ell = (1, \ell_2, \dots, \ell_\Lambda)$, where $\ell_1 = 1$ and other lengths are arbitrary. Then there exists a lattice path in height representation $\mu$ and a mode permutation $\sigma \in S_m$ such that the space of possible output Fock states resulting from applying $X$ to the input $\ket{\vec n}$ is a PCS $\gen{\downset\mu}^\sigma$ with
  \[ \mu_{\sigma(a)} = \sum_{b} n_b \]
  where the summation runs over input modes $b$ that are reachable from output mode $a$ (see \Cref{sec:progressive-decomposition}).
  The permutation $\sigma$ ensures $\mu$ is nondecreasing.
\end{corollary}

\subsubsection{Discussion}

The tools above allow us to fully and efficiently describe the output state space of any system of loops where $\ell_1 = 1$, on any input basis state, and any number of modes. This includes even large systems where the state space would be too big for a computer memory, since in our formalism, we need only store two objects: $\mu$ and $\sigma$. The condition that $\ell_1 = 1$ stems from the fact that our evolution rule in \Cref{thm:PCS-BS-evolution} requires its input space to be a PCS: from an arbitrary input (not a PCS), in particular a single number basis state, we obtain a PCS by applying a loop $L_1$ (see \Cref{sec:loops-give-downsets}).
However, requiring that $\ell_1=1$ is reasonable in the context of time-bin loop-based systems, and it is always the case in the present paper.

We use the lattice path formalism to count the dimension of the state space characterized by a downset. Obtaining a general closed-form formula that counts any downset is difficult, though possible when the loop lengths are restricted (e.g. Br\'adler \& Wallner proved a closed-form formula when all lengths are $\ell_j=1$~\cite{bradler_wallner_2021}). However, in \Cref{app:counting-lp}, we provide a dynamic programming algorithm that can efficiently count lattice paths in general systems without such restriction.

The applicability of \Cref{thm:PCS-BS-evolution} goes beyond staircase-shaped circuits, as it allows us to characterize the evolution of the state space beamsplitter-by-beamsplitter, particularly in the order given by the progressive decomposition (see \Cref{sec:progressive-decomposition}). The last ingredient missing to fully describe the state space throughout all of the steps of the progressive algorithm are partial mid-circuit measurements, which we develop in the next section.

\subsection{Measuring arbitrary modes}
\label{sec:meas-arbitr-modes}

In the progressive simulation algorithm, after evolving the statevector through a component $P_a$ of the progressive decomposition, we perform a measurement of its top mode, after which it is removed from the simulation. To fully describe the state space at every step during the simulation, we need to capture this partial measurement in our lattice path formalism. Even though we always measure the top mode, due to the permutation of a PCS this may be represented by a different mode in the formalism. We thus adapt the formalism so that it  can encode the measurement of an arbitrary mode of a cumulative space.

Our key contribution in this section is showing that a state space described by lattice paths can still be described by lattice paths after a measurement is performed. Specifically, we show the measurement of $x$ photons in mode $a$ sends a maximal path $\mu$ over $m$ modes and $n$ photons to the new maximal path $\mu^{(n_a = x)}$ over $m-1$ modes and $n-x$ photons.

We follow three steps, shown in a running example in \Cref{fig:lp-measurement}: First, we encode the projection to a subspace where each basis state has $n_a=x$ (\Cref{sec:projection}). This is not a single downset anymore, but it is made of a collection of regions in the lattice, each defined by $\lambda_{a-1}$, the number of photons in modes $\iicc{0,a-1}$. Second, we delete the measured mode $a$ and the $x$ photons inside by squashing the lattice diagram (\Cref{sec:contraction}) which is vital for merging the aforementioned regions back together. The final step is the merge itself: we show that taken together, the squashed regions give a new PCS, and hence our efficient representation can encode measurements (\Cref{sec:merge}).

\begin{figure}[htp]
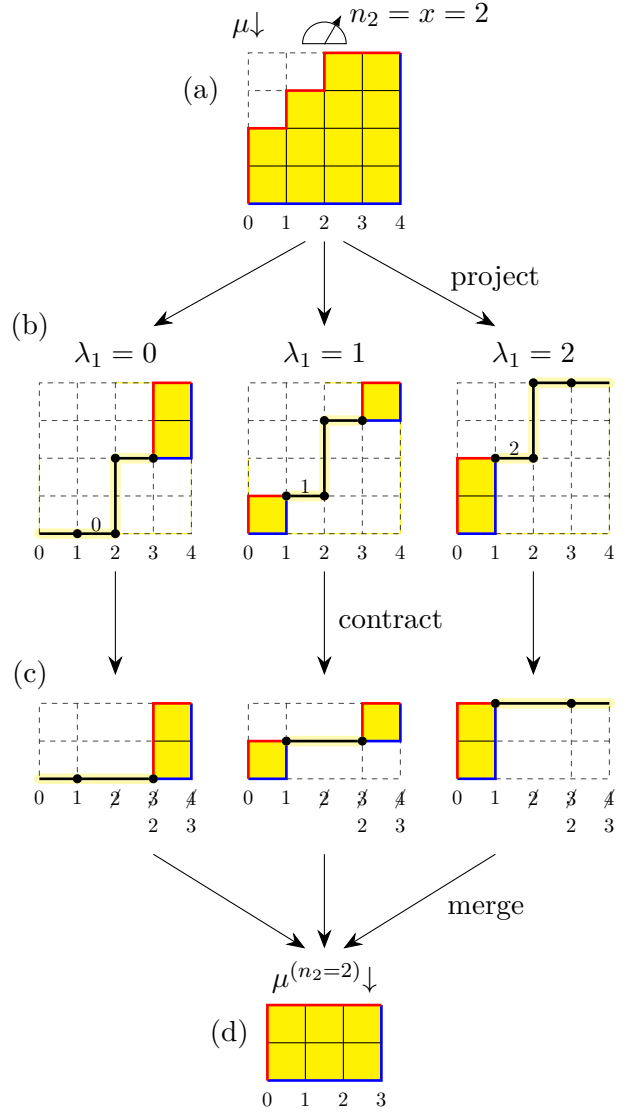

  \ctikzfig{measurement-example-projection}
  \caption{Encoding a measurement outcome in the lattice path formalism.
    \captionitem{a}~State space before a measurement which is a cumulative space $\gen{\downset\mu}$ with $\mu = (2,3,4,4,4)$. We measure $x = 2$ photons in mode $a = 2$.
    \captionitem{b}~First, we \emph{project} to a subspace compatible with the measurement, keeping only those lattice paths $\lambda$ that have $n^\lambda_a = x$. The set of these paths is, in general, not a single skew diagram in the lattice, but a union of several skew diagrams $\{ \mu^q/\nu^q \}_q$ where $q = 0, \dots, q_{\max}$. Each $\mu^q/\nu^q$ represents the subspace with $q$ photons in modes $<a$. These constraints give $\mu^q/\nu^q$ the S-shaped region centered at mode~$a$, starting at height $q$, and with step $x$.
    \captionitem{c}~Next, we delete the measured mode and the photons within. We implement this as a \emph{contraction} of the S-shaped region: squash the diagram vertically to ensure that at mode $a$, the step has size zero. Then erase the vertical line that corresponds to mode $a$, merging the boxes to the left and to the right of it. Relabel the modes to the right of $a$ as $b \mapsto b-1$.
    \captionitem{d}~Finally, we \emph{merge} the contracted diagrams. This results in a contiguous region in the lattice, the downset $\downset{\mu^{(n_a = x)}}$.
  }
  \label{fig:lp-measurement}
\end{figure}

\subsubsection{Projection}
\label{sec:projection}

When we measure $x$ photons in mode $a$, we project the statevector to the subspace that is compatible with that measurement. This is the step (a) $\to$ (b) in the running example in \cref{fig:lp-measurement}. We only keep those lattice paths that have a step of size $x$ in mode $a$:
\begin{equation}
  \label{eq:measurement-projection-single}
  \{ \lambda \in \downset\mu : n^\lambda_a = x \}.
\end{equation}
\noindent The above is generally not the set of paths in some downset, or even a contiguous region in the lattice: the projection removes certain paths, creating gaps in the region. This is a problem for our representation, since this would no longer allow us to efficiently characterize very large state spaces using a single maximal path $\mu$.

The solution is to decompose (and later recompose) the above set \eqref{eq:measurement-projection-single} into lattice regions defined by the number of photons in modes $\iicc{0,a-1}$, encoded by restricting the paths $\lambda$ to have height $\lambda_{a-1} = q$ (see \eqref{eq:n-to-lambda}), for each possible choice of~$q$. These regions, depicted in \Cref{fig:lp-measurement}b, correspond to a collection of disjoint subspaces of the state space, and are represented by regions of the lattice that can be efficiently described by their boundaries, i.e. the minimum and maximum paths. We call them \emph{skew diagrams} $\mu^q/\nu^q$ where $q$ is the choice of $\lambda_{a-1}$. Where a downset $\downset\mu$ includes all lattice paths in the region below $\mu$, a skew diagram $\mu^q/\nu^q$ contains paths between and including $\mu^q$ and $\nu^q$:
\begin{equation}
  \label{eq:skew-diagram-interval}
  \mu^q/\nu^q = \{ \lambda : \nu^q \le \lambda \le \mu^q \}.
\end{equation}
In \Cref{fig:lp-measurement}, we draw $\mu^q$ \red{red} and $\nu^q$ \blue{blue} wherever they differ. Notice that they all include an S-shape that implements the step of size $x$ in mode $a$ starting at height $q$.
The value of $q$ ranges between $0$ and $q_{\max} = \min\{ \mu_{a-1}, \mu_a - x \}$, the latter easily found as the highest $y$ coordinate in the diagram where the start of the S-shape of height $x$ can be placed within a diagram $\downset\mu$.

Informally, the systematic way to obtain all $\mu^q/\nu^q$ is to start with $q=0$, place the S-shape at the corresponding position in the lattice, and restrict the yellow region to be compatible with the S-shape. Then increase $q$, sliding the S-shape upward, and repeat the procedure until we reach the top of the lattice ($q = q_{\max}$). We show the formal details of computing all skew diagrams $\mu^q/\nu^q$ in \Cref{app:lp-projection}.

\subsubsection{Contraction}
\label{sec:contraction}

The next step is to delete the measured mode $a$ and the photons it contained from all the skew diagrams produced in the previous step. To do this, we remove the $x$ measured photons by removing $x$ rows from each diagram (i.e. remove boxes between vertical positions $q$ and $q+x$), then remove one column at position $a$. This implements a contraction of the S-shaped region in $\mu^q/\nu^q$ which is shown as the step (b) $\to$ (c) in the running example in \Cref{fig:lp-measurement}. We are left with contracted skew diagrams that represent subspaces of a system of $m-1$ modes and $n-x$ photons. We provide a formal presentation in \Cref{app:lp-contraction}.

If the state space is a PCS with a nontrivial permutation $\sigma$, we implement the contraction on $\sigma$ by deleting the input $a$, output $\sigma(a)$, and shifting the labels of modes following the deleted modes by $-1$. We show an example in \Cref{fig:permutation-contraction}.

\begin{figure}[t]
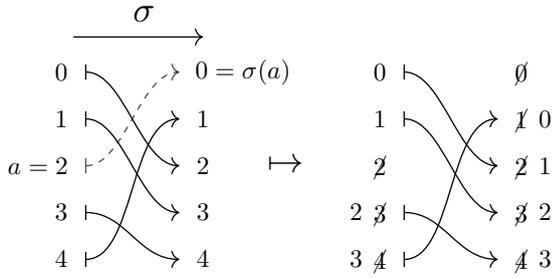

  \fwtikzfig{permutation_contraction}
  \caption{Example action of the contraction of mode $a=2$ on the permutation $\sigma$ of a PCS. On the LHS, we delete input mode $a$ on the left, as well as its target, the output mode $\sigma(a)$ on the right, related by the dashed arrow. All the modes below these, on each side, have their labels shifted by $-1$. We obtain the contracted permutation on the RHS.}
  \label{fig:permutation-contraction}
\end{figure}

\subsubsection{Merge}
\label{sec:merge}

Finally, we merge all of the contracted skew diagrams, one for each $q=0,\dots,q_{\max}$, to again obtain a representation where a single lattice path characterizes the whole state space as a downset. We show this in the step \mbox{(c) $\to$ (d)} of the running example in \Cref{fig:lp-measurement}.

Informally, to see why this yields a downset, consider the following: After contraction of the S-shapes, each skew diagram encodes only the choice of $\lambda_{a-1} = q$, that is, it enforces that the paths of that diagram pass the point $(a-1, q)$. However, the paths are otherwise free in the other modes, as long as they are consistent with the former constraint, as well as the original constraints imposed by $\mu/\nu$; see \Cref{fig:lp-measurement}c. We can imagine ``scanning'' the contracted lattice, setting $q = 0$, and writing down all the paths that satisfy the constraint $\lambda_{a-1}=q$, then advancing to $q = 1$, etc. Note in particular that we have a skew diagram for each $q \in \iicc{0,q_{\max}}$ without gaps. This captures all the valid lattice paths after a measurement, as seen in the example in \Cref{fig:lp-measurement}d. However, this part is formally nontrivial, and we provide more detail in \Cref{app:lp-merge}.

\subsubsection{Full description of measurement}

Putting the three steps together, we obtain one of our main results. It gives us the ability to efficiently encode the effect of a measurement of an arbitrary mode, with a given measurement outcome, on the state space using the lattice path formalism:

\begin{theorem}[measurement]
  \label{thm:lp-measurement}
  Starting with a cumulative space $\gen{\downset\mu}$, the state space left after measuring $x$ photons in mode $a$ is a cumulative space $\negen{\downset{\mu^{(n_a = x)}}}$ where the maximal path $\mu^{(n_a = x)}$ is the contraction of $\mu^{q_{\max}}$. The latter is the maximum path of the skew diagram corresponding to $q = q_{\max} = \min\{ \mu_{a-1}, \mu_a - x \}$.
\end{theorem}
The proof can be found in \Cref{app:lp-merge}.

\subsection{Tracking the state space in progressive simulation}
\label{sec:tracking-state-space}

Together, the maximal lattice path $\mu^t$ and the mode permutation $\sigma^t$ at step $t$ fully describe the state space throughout the simulation, where each step is either an individual beamsplitter or a measurement. If step $t$ is a beamsplitter, then we obtain $\mu^{t+1}$ and $\sigma^{t+1}$ using \Cref{thm:PCS-BS-evolution}. Note that if the beamsplitter is part of the initial $L_1$, as described in \Cref{sec:coupling-new-modes}, we must also append a new mode to the PCS. On the other hand, if step $t$ is a measurement, we choose its outcome number of photons, and obtain $\mu^{t+1}$ using \Cref{thm:lp-measurement} and $\sigma^{t+1}$ as described in \Cref{sec:contraction}.

Our extension of the lattice path formalism allows us to determine the size of the state space that the progressive simulator must keep in memory to reach a specific output measurement $\ket{\vec n'}$. However, this does not tell us how $\ket{\vec n'}$ should be chosen. To quantify metrics such as the average-case memory complexity of simulating a given architecture, we would need to sample $\ket{\vec n'}$ from a generally hard boson sampling distribution. That would be counterproductive because we want to estimate the complexity of sampling from this distribution in the first place.
In \Cref{sec:quantifying-complexity}, we show how to solve this issue and thus use the lattice path formalism to quantify the complexity of the progressive algorithm itself.

\section{Quantifying the simulation complexity of loop-based interferometers}
\label{sec:quantifying-complexity}

Now that we have a tool to characterize the state space at every point in our simulation, including individual beamsplitters as well as measurements, we are almost ready to use this to estimate the complexity of simulating loop-based systems, defined in terms of the memory requirement:

\begin{definition}[memory complexity]
  \label{def:memory}
  Given a loop-based system with loop lengths $\vec\ell$, total number of modes $m$, and an input state $\ket{\vec n}$, we define a function $M$ that captures the memory required to run a simulation of that system producing the output pattern $\vec n'$:
  \[ M(\vec n') = \max_t |\downset{\mu^t}| \]
  where $\mu^t$ is the maximal lattice path at step $t$ of the simulation.
\end{definition}

This metric captures the maximum memory needed to store the state vector over the course of the simulation. We note that $M$ does not assume a choice of strong simulator or how the data is represented: it is simply the number of nonzero coordinates to be stored. Since the number of photons $n'_a$ measured in mode $a$ at any given time affects the size of the state vector after the measurement, the value of $M$ does depend on the choice of $\vec n'$.

As this expression is defined in terms of a single output $\vec n'$, to quantify the complexity of simulating an experiment, we must estimate the statistics of $M(\vec n')$ over a range of possible values of $\vec n'$. One useful metric is the average complexity of an experiment:
\begin{equation}
  \label{eq:expectation-M-p}
  \bar M \deq \E_{\vec n' \sim p} \left[ M(\vec n') \right]
  = \sum_{\vec n'} p(\vec n') M(\vec n'),
\end{equation}
This is the expectation over all possible samples $\vec n'$ taken with respect to the physical probability distribution $p$, i.e. the actual outcome probability. 

The distribution $p$ is hard to compute or sample from for all but the smallest systems, so we need a tractable way to evaluate the expectation \eqref{eq:expectation-M-p}. To solve this issue, in \Cref{sec:good-heuristic}, we build a heuristic probability distribution that is easy to sample from and can be used to predict the memory complexity. We validate it by comparison against samples from small and medium-size architectures that can be simulate. Then in \Cref{sec:complexity-vs-modes,sec:comparison-power-law-archs}, we use the heuristic to study the behavior of complexity with respect the total number of modes and the lengths of loops in large systems in order to assess their potential for quantum advantage.

\subsection{A good heuristic probability distribution}
\label{sec:good-heuristic}

To evaluate the expectation $\bar M$ in \cref{eq:expectation-M-p} without having to sample from the true distribution $p$ for a large system, we propose instead using a heuristic probability distribution $p_H$ that is easy to sample from and provides a good estimate of $\bar M$. Given the same input state and loop structure that gave rise to $p$, we now assume that all possible outcomes $\vec n'$ are equally likely, thus defining $p_H$ as the uniform probability:
\[
  p_H(\vec n') = \frac1{\text{number of possible outcomes}}.
\]

There is a simple and efficient way to generate samples $\vec n'$ from the distribution $p_H$ mode-by-mode using the progressive decomposition:
\begin{lemma}
  \label{lem:heuristic-marginals}
  Suppose $\mu$ is a maximal path after simulating the circuit up to the end of component~$P_a$, just before the measurement of mode $a$. In particular, $\mu$ already encodes the previously obtained measurement outcomes $n'_0, \dots, n'_{a-1}$. Then the marginal heuristic probability of measuring $x$ photons in mode~$a$ is
  \begin{equation}
    \label{eq:heuristic-chain-rule}
    p_H(n'_a = x \mid n'_0, \dots, n'_{a-1})
    = \frac{| \downset{\mu^{(n_a = x)}} |}{\left| \downset{\mu} \right|},
  \end{equation}
  where previous measurements $n'_0, \dots, n'_{a-1}$ are fixed, the value $x \in \iicc{0, \mu_a}$, and $\mu^{(n_a = x)}$ is the maximal lattice path after the measurement.
\end{lemma}
\begin{proof}[Sketch proof]
  The maximal path $\mu$ already encodes the outcomes $n'_0, \dots, n'_{a-1}$ and the probability of outcomes is uniform, so we just need to count the paths in $\downset\mu$ with $n_a = x$.
\end{proof}

Note that the uniform heuristic $p_H$ is very different from the true SPBS distribution $p$, but is a good first choice because it requires no information about $p$. We show in \Cref{sec:validation-heuristic} that $p_H$ does in fact provide a good estimate of the statistics of the memory complexity, giving similar results to the real distribution.

\begin{figure*}[!ht]
  \includegraphics[width=\linewidth]{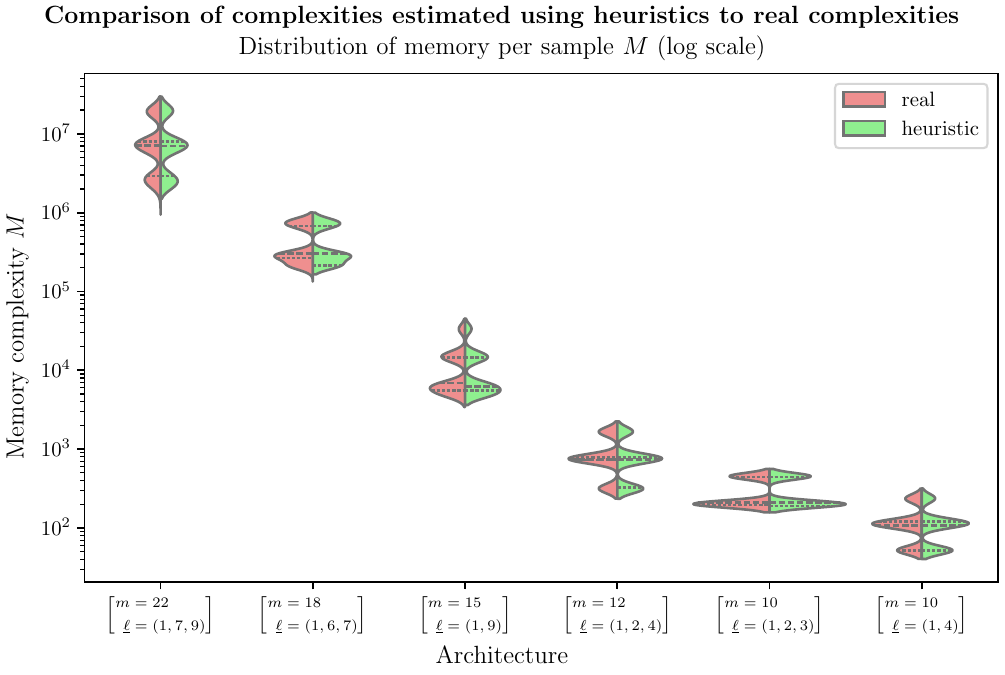}
  \caption{Validation of the heuristic probability distribution $p_H$. The horizontal axis shows different architectures $(m, \vec\ell, \ket{\vec n}$) where $\ket{\vec n} = \ket{1,0,1,0,\dots}$ with $\lceil \frac m2 \rceil$ photons. The vertical axis shows the memory requirement to simulate samples from a given architecture using our method, and is scaled logarithmically. The data are shown as mirrored violin plots, showing the distribution of memory with respect to true probability $\bar p$ in red on the left, and with respect to the heuristic $p_H$ in green on the right. Inscribed in the violins are the quartiles: median dashed, $25$-th and $75$-th percentiles dotted.
  }
  \label{fig:heuristic-validation}
\end{figure*}

\subsubsection{Validation of the heuristic}
\label{sec:validation-heuristic}

We empirically validate on small, easy-to-simulate circuits that the heuristic $p_H$ can be used to compute the average memory complexity. Since a real boson sampling distribution $p$ depends on the choice of beamsplitter parameters $(\theta_i)_i$, in our numerical experiments we average over several choices of $(\theta_i)_i$ drawn uniformly at random. Our validation consists of showing the following approximation:
\begin{equation}
  \label{eq:heuristic-approximates-real}
  \E_{\vec n' \sim p_H} [ M(\vec n') ]
  \approx
  \E_{\substack{(\theta_i)_i \\ \vec n' \sim p}} [ M(\vec n') ]
\end{equation}
In the following, we call $\bar p$ the probability distribution on the RHS above which averages over~$(\theta_i)_i$.

We choose, here and in the following sections, the input state $\ket{\vec n} = \ket{1,0,1,0,\dots}$ with $\lceil \frac m 2 \rceil$ photons. This choice is motivated by two considerations: First, current experimental efforts operate roughly in this regime, for example~\cite{he_time-bin-encoded-spbs_2017} with input similar to ours, or~\cite{Wang_boson-sampling-20-input_2019} with a slightly different ratio between the numbers of photons and modes. This regime is interesting in near-term experiments because there is a large space of possible measurement outcomes even without perfect photon number resolution. Second, recent theoretical research supports the investigation of this regime (and generally $m = \mathcal{O}(n)$), showing that it is hard to simulate classically~\cite{bouland_complexity-theoretic_2023}.

As shown in \Cref{fig:heuristic-validation}, we find over several choices of loop architectures that the distribution of memory complexities using heuristic $p_H$ closely matches the distribution obtained using~$\bar p$. For visual convenience, we only show a small selection of experiments represented along the horizontal axis. The vertical axis represents the memory complexity $M(\vec n')$, and is scaled logarithmically. We take $N=1000$ samples from each the true distribution $\bar p$ and the heuristic $p_H$, and from these samples we estimate the corresponding probability distributions of $M$ using \emph{kernel density estimation (KDE)}.\footnote{KDE uses an important parameter called \emph{bandwidth}, and we use the \emph{Scott's rule}~\cite{turlach_bandwidth_1999} to automatically select it.} We plot these estimates as mirrored violin plots: true distribution on the left in red, and heuristic on the right in green. For further comparison, each violin is inscribed with the quartiles: dashed line for median and dotted lines for $25$-th and $75$-th percentiles.

The results shown in \Cref{fig:heuristic-validation} confirm that the uniform heuristic $p_H$ can be used to give good predictions for the memory complexity. This means we can use it to reason about systems where generating samples using the real distribution is not possible.

\subsection{Complexity of a large loop-based system}
\label{sec:complexity-vs-modes}

\begin{figure*}[t]
  \includegraphics[width=\linewidth]{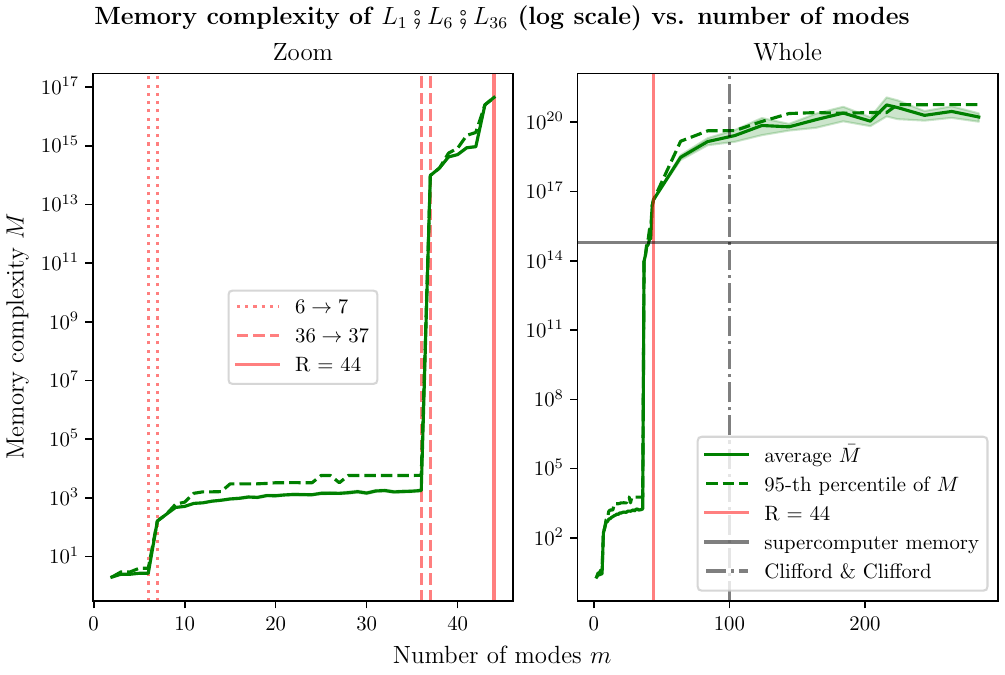}
  \caption{Memory complexity of simulating a loop-based system of lengths $\vec\ell = (1,6,36)$ with a varying number of modes $m$ between $2$ and $284$ along the horizontal axis. We consider an input state of $\ket{1,0,1,0,\dots}$ where there are half as many photons as modes. On the vertical axis, we show the memory complexity estimated using the heuristic~$p_H$. We show two subplots of the same data, where the LHS picture only shows the range of $m$ between $2$ and $44$ to make this region more visible. In each plot, we show the mean memory $\bar M$ as a full line and the $95$-th percentile of the samples of memory as a dashed line. We observe that the memory increases abruptly between three regions, corresponding to each loop entering the simulation. We also observe that the complexity starts to saturate after $m = R = 44$, shown by the full red line. For comparison, we plot an estimate of the total memory available to a state of the art supercomputer (horizontal grey line), assuming that each coordinate is encoded as a double-precision complex number (16 bytes). We highlight the number of modes ($m = 100$ with our input) up to which the Clifford \& Clifford algorithm can simulate the circuit (vertical dash-dotted line).}
  \label{fig:modespace-experiments}
\end{figure*}

Now that we have the tools to estimate the memory complexity of simulating arbitrary loop-based systems with our progressive algorithm, we can use them to analyze the largest loop-based systems built to date. We focus on a system with loop lengths $\vec\ell = (1, 6, 36)$ used by Madsen\etal in their \emph{Borealis} GBS experiment~\cite{madsen_quantum_2022}, and we study the complexity of this architecture in the context of SPBS as a function of the total number of modes $m$. To do this, we replace the original Gaussian input with $\ket{\vec n} = \ket{1,0,1,0,\dots}$.

In \Cref{fig:modespace-experiments}, we plot the memory complexity as a function of the total number of modes $m$ between~$2$ and $284$. For each value of~$m$, we take $N = 1000$ samples of memory complexity using the heuristic probability $p_H$. We study the average memory $\bar M$ (full line in the figure), as well as the $95$-th percentile of the samples of~$M$ (dashed line). The latter represents the memory required to obtain almost all possible samples and we use it as a proxy for quantum advantage: if our classical computer has more memory than the $95$-th percentile, then we can successfully run the progressive simulator at least $95\ \%$ of the time. 

In the figure, we draw a comparison between the memory complexity and the memory available to a state of the art supercomputer, shown as a horizontal gray line. Our chosen value of $10^{15}$ is inspired by the OLCF-5 (Frontier) system~\cite{noauthor_frontier_nodate}, taking a loose assumption that all memory within the system can be used to store the quantum state. This overestimates the real capacity, but is independent of any implementation details. For comparison with the values $M(\vec n')$, we assume each nonzero coordinate is stored as a double-precision complex float.

Observe that there are three different regions in terms of $m$ where the memory complexity has different behaviour, and there are abrupt jumps between these. The regions correspond to values of $m$ where different numbers of loops are active. For example, if $m \le 6$, then $L_1 \then L_6 \then L_{36}$ effectively becomes a single loop system $L_1$, because there are not enough modes for the longer loops to have any effect. However, as $m$ increases to~$7$, the second loop becomes active, and we see a jump in complexity. A similar situation repeats for the transition from $m = 36$ to $37$, where the third loop becomes active.

The memory complexity starts to saturate once the total number of modes reaches the value $R$. Recall that $R$ is the maximum number of modes in each component $P_a$ of the progressive decomposition, in this case $R = 44$. In the region of $m$ slightly larger than $R$, the components $P_a$ are reaching their maximum possible size and the complexity reaches a plateau shortly after. We gain little additional complexity by increasing $m$ further, though for the system with $\ell = 6$, already $m \approx R$ is in the regime where our progressive algorithm fails even on a supercomputer.

\subsection{Comparison of power-law architectures}
\label{sec:comparison-power-law-archs}

\begin{figure*}[t]
  \includegraphics[width=\linewidth]{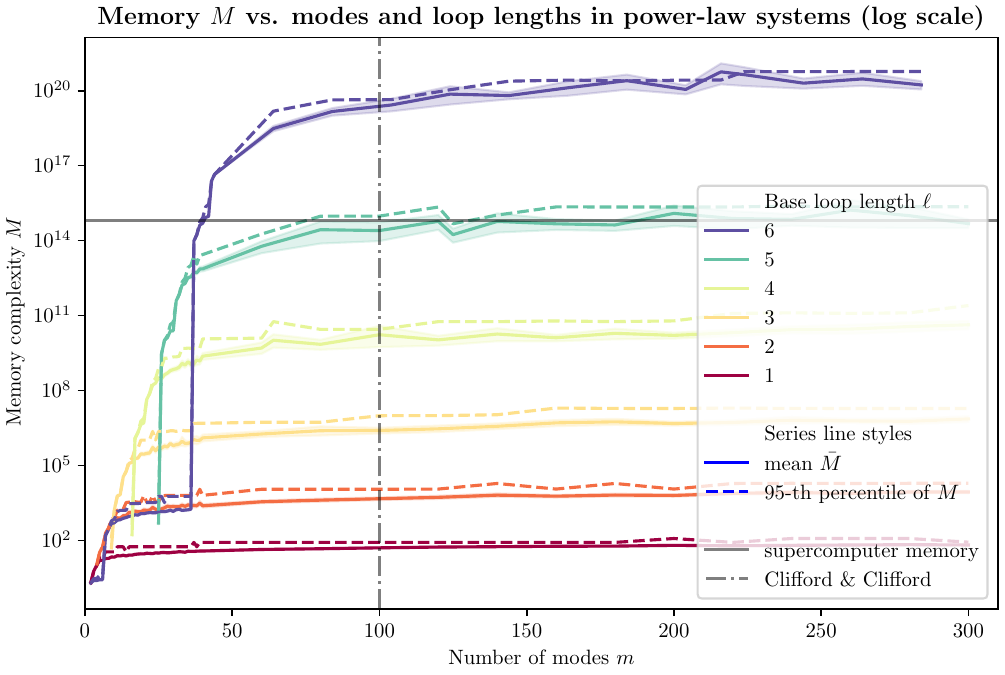}
  \caption{Comparison of power-law architectures $\vec\ell = (1, \ell, \ell^2)$ for different values of the base loop length $\ell$, signified by the colors of the lines. We consider an input state of $\ket{1,0,1,0,\dots}$ where there are half as many photons as modes. The plot is similar to \Cref{fig:modespace-experiments}: we show the total number of modes on the horizontal axis, the memory complexity on the vertical axis, and we plot the mean $\bar M$ and $95$-th percentile, comparing them to the supercomputer memory drawn as horizontal grey line. We observe that $\ell=5$ is a threshold for complexity when using the progressive simulator. We furthermore highlight the number of modes ($m = 100$ with our input of $n = 50$ photons) up to which the Clifford \& Clifford algorithm succeeds in simulating the circuit. Thus we divide the plot into quadrants representing regimes where different simulators are appropriate, and we observe the bottom right quadrant containing systems simulable by our algorithm, but not by the Clifford \& Clifford algorithm.}
  \label{fig:comparison-power-law}
\end{figure*}

We generalize our analysis to three-loop power-law architectures from \Cref{def:power-law} where the base length $\ell$ may be different and we study how this affects the complexity.
Similarly to the previous section, we also vary the total number of modes $m$, and we use the input state $\ket{\vec n} = \ket{1,0,1,0, \dots}$ with $\lceil \frac m2 \rceil$ photons. We again compare the memory complexity to the capacity of a state-of-the-art supercomputer to better contextualize the potential for quantum advantage.

We present the results in \Cref{fig:comparison-power-law} where we show several colored series of data, each corresponding to a different value of $\ell$. Similarly to the previous \cref{fig:modespace-experiments}, we show the average and $95$-th percentile memory as functions of $m$.
We find that three-loop power-law systems with base length $\ell < 5$ and the above input are classically simulable using a large enough computer, even for a large number of modes. Thus we consider them not complex enough to reach quantum advantage.

When $\ell=5$, we reach a threshold in complexity. Observe in \Cref{fig:comparison-power-law} that the memory requirement to simulate a system with $\ell=5$ with $\lceil \frac m2 \rceil$ photons using our progressive simulator just surpasses the capacity of a state-of-the-art supercomputer. We see this happen when $m \ge 80$.

It follows that $\ell=5$ is a lower bound for quantum advantage: Systems with $\ell < 5$ are clearly simulable by the progressive simulation algorithm. However, while systems with $\ell \ge 5$ are not simulable by our algorithm, they may still be simulable using other methods. For instance, in addition to having a sufficiently complex loop structure a quantum system must also have roughly 40 to 50 photons at least to avoid classical simulability with the algorithm of Clifford and Clifford~\cite{Cliffords_2017}. We also highlight this threshold, which falls on $m=100$ modes for our input $\ket{1,0,1,0,\dots}$, in \Cref{fig:comparison-power-law}. Moreover, as we discuss in the next section, experimental imperfections such as photon loss and distinguishability can reduce the requirements for classical simulation.

\section{Discussion and conclusion}
\label{sec:discussion}

Here, we discuss ways in which our classical simulation method could be extended to broader settings beyond ideal boson sampling with single photons.

\paragraph{Extension to loss}
Here, we show that this method can be extended to simulate lossy circuits. Loss on a mode can be modeled as a beam splitter between the mode and the environment. We can thus evolve an initially $n$-mode pure state through loss on mode $a$ by applying a beam splitter between mode $a$ and the environment. We then simulate a measurement process on the environmental mode, by calculating the probabilities of measuring different photon numbers on the environmental mode and sampling from these possible outcomes. This measurement collapses the initial $n$ modes back to a pure state with potentially fewer photons. Though the probabilistic outcome of the environmental measurement makes this process stochastic, the resulting statistics match those obtained by tracing out the environmental mode and using the density matrix formalism. However, our method avoids the quadratic overhead stemming from the use of a density matrix. Repeating this process at every loss location in the circuit allows us to exactly simulate sampling from the lossy state using our classical simulation algorithm.

Moreover, the complexity of the lossless case provides an upper bound to the complexity of simulating a lossy circuit. Adding additional beam splitters to simulate loss adds an overhead which is only linear in the number of optical components. However, the state space is also significantly smaller than the lossless case since we now simulate a state with fewer photons. Though our lattice path formalism for determining the exact size of the state space cannot be directly applied in the case of loss, the lossless case thus provides an upper bound.

\paragraph{Extension to distinguishable photons}
Our method can also be extended to simulating boson sampling with partial distinguishability. Following~\cite{renema2018efficient, sparrow2018quantum}, a simple model for a system with partially distinguishable photons involves considering that each photon is in a superposition of being in an indistinguishable mode and a distinguishable mode. Sampling from this model can be simulated by randomly assigning the input photons into a distinguishable set and an indistinguishable set at the beginning of the simulation, with an assignment probability determined by the initial superposition state. Evolving the indistinguishable photons through the circuit and collecting the output measurements, then simulating the distinguishable photons and summing both measurements, would yield the correct output statistics. The complexity of this method is dominated by the complexity of simulating this subset of indistinguishable photons, which is upper-bounded by the complexity of simulating the full set of indistinguishable photons given by our lattice path formalism. We note that more involved models for partial distinguishability have been developed, which we anticipate could also be used within our simulation method \cite{shchesnovich2015partial,annoni2025incoherent}.

\paragraph{Gaussian Boson Sampling}

We focus on SPBS, but there are other variations of Boson Sampling that have been investigated, such as Gaussian Boson Sampling (GBS). For example, in \Cref{sec:complexity-vs-modes}, we studied the system of loops $\vec\ell = (1,6,36)$ which was inspired by a GBS advantage experiment by Madsen\etal\cite{madsen_quantum_2022}. GBS is underpinned by a different formalism than SPBS, so our methods are not directly transferable. However, it would be interesting to explore whether the staircase structure of the interferometer (\Cref{fig:big_picture:circuit}) could be exploited to accelerate GBS simulations.

\subsection{Summary and outlook}

In this work, we study the potential for quantum advantage in loop-based boson sampling systems, and to this end, we provide new tools to quantify their classical simulation complexity. First, we exploit the circuit topology to design a new simulation method that requires that only small subcircuits are strongly simulated. This can use any wavefunction Fock basis simulator, which implies that a relevant choice of complexity metric is the memory required to store the state, i.e. the dimension of the reachable state space.

To fully and efficiently characterize the reachable state space, we extend the lattice path formalism of Br\'adler \& Wallner~\cite{bradler_wallner_2021}. Our extension can encode the action of beamsplitters between arbitrary modes, and similarly the outcomes of measurements of arbitrary modes. To use the formalism to predict the memory complexity of simulating large systems, we define a heuristic probability on the possible outcomes that is easy to sample from. We validate the statistics given by the heuristic by comparing its predictions to empirical data from simulations of small to medium-scale systems.

Using these tools, we study the complexity of power-law architectures where the loops have lengths $1, \ell, \ell^2$ and we use input state $\ket{1,0,1,0,\dots}$, to represent the common experimental scenario where there are half as many photons as modes. A system with $\ell < 5$ (and any number of modes $m$) can be efficiently simulated using our method, and thus does not achieve advantage. We provide a lower bound for advantage in such systems at $\ell \ge 5$.

To conclude, our work provides new insight into the boundaries for quantum advantage using single photons in loop-based time-bin interferometers. We develop new tools that provide new lower bounds for the arrangement of the loops. We anticipate that this work will help guide future experimental implementations of boson sampling, and points the way for further research into quantum advantage with loop-based systems under different experimental scenarios.

\section*{Acknowledgements}

We would like to thank the members of the team at ORCA for helpful discussions. This work was supported by the Quantum Data Centre of the Future, project number 10004793, funded by the Innovate UK Industrial Strategy Challenge Fund (ISCF).
R.G.-P. was supported by the EPSRC-funded project Benchmarking Quantum Advantage.

\appendix

\section{Results on the loop-based structure}
\label{app:loop_theorems}

To formally prove \Cref{thm:relevant-modes}, we formalize the notion of travelling through a circuit $X$ backwards. For an output mode $a$, we denote $X^\dagger(a)$ the set of all input modes, photons from which can reach the output $a$ with nonzero probability, constrained only by the circuit topology. That is, the set $X^\dagger(a)$ is a preimage of the relation encoding connectivity from input to output modes.

We recall the statement of the theorem below:

\noindent
{\normalfont\bfseries Theorem\ \ref*{thm:relevant-modes}\ {\normalfont (relevant modes)}.}
In a progressive decomposition of a loop-based system $X$ made of loops of lengths $\vec\ell = (\ell_1, \dots, \ell_\Lambda)$, each component $P_a$ acts non-trivially on at most the number of \emph{relevant} modes
\begin{equation}
  \tag{\ref*{eq:relevant-modes}}
  R = 1 + \sum_{i = 1}^\Lambda \ell_i.
\end{equation}

\begin{proof}
  Recall that a component $P_a$ is made of beamsplitters relevant to measuring output mode~$a$. We use these beamsplitters to reach the modes in $X^\dagger(a)$ while traversing the circuit backward: a beamsplitter allows us to jump to a different mode, as shown in the example of \Cref{fig:preimage-of-0}.

  To bound the number of relevant modes, we need to identify the furthest reachable mode, i.e. the one identified by the maximum number $\max X^\dagger(a)$.
  The mode $\max X^\dagger(a)$ is found by starting at output $a$ and following backward the circuit $X = L_1 \then L_{\ell_2} \then \cdots \then L_{\ell_\Lambda}$, taking a step $+\ell_j$ (towards a mode identified by a higher number) at each loop $L_{\ell_j}$; for example, see the red path with arrows $\tikzredarrowsback$ in \Cref{fig:preimage-of-0}. This way, reach the maximum mode: $\max X^\dagger(a) = a + \sum_{j=1}^\Lambda \ell_j$.

  Finally, recall that by construction of the progressive decomposition, the component~$P_a$ does not act on modes $\iicc{0, a-1}$: Any beamsplitters affecting these are captured by $P_0, \dots, P_{a-1}$, and furthermore, the modes $\iicc{0,a-1}$ are measured before the component~$P_a$ is simulated. Having previously identified the maximum reachable mode $\max X^\dagger(a)$, we conclude that the maximum number of modes that $P_a$ can act on is:
  \[
    \left| \iicc{a, a+{\textstyle \sum_{j=1}^\Lambda \ell_j}} \right| = 1 + \sum_{j=1}^\Lambda \ell_j.
    \qedhere
  \]
\end{proof}

Finally, we note that this upper bound is saturated whenever a component $P_a$ starts with a staircase of nearest-neighbour beamsplitters given by $L_1$ -- namely, this is the case for $P_0$ in a system with $\ell_1=1$. The reason is that the set of inputs reachable via $L_1$ from any output mode~$b$ is $L_1^\dagger(b) = \iicc{0, b+1}$. This is because the staircase of nearest neighbour beamsplitters allows us to take one step from $b$ to $b+1$, and as many steps towards lower ($< b$) modes as we like, as shown in \Cref{fig:preimage-of-0} by the yellow region.

Other progressive components ($P_a$ for $a>0$) in a system with $\ell_1=1$, and more generally any component if $\ell_1 \ne 1$, may act on fewer modes than $R$. The procedure of travelling through the circuit backward, and using beamsplitters to take steps up or down, yields a set which is not an interval, but has gaps. We do not need more detail on this for our work, however, the set $X^\dagger(a)$ can be exactly characterized using the rich theory of numerical monoids, see~\cite{rosales2009numerical}.

\section{Formal details of the lattice path formalism}
\label{app:lattice-path-formalism-details}

In order to give a more formal proof of \Cref{thm:lp-measurement}, particularly the part where we merge contracted skew diagrams, we need to fill in some details omitted in \Cref{sec:loops-give-downsets,sec:a-new-formalism}. We need tools from order theory: we present generalities inlined in the text, and we only number definitions or lemmas specific to our work.

\subsection{Young order}

We start by recalling that a \emph{partial order} $\le$ on a set $P$ is a generalization of a total order (e.g.~the usual ordering on the real line~$\R$) where some elements may be incomparable. The set $P$ is a \emph{partially ordered set (poset)}; for more details, see~\cite{stanley_enumerative_2012}.
The order relevant to us is the following:

\begin{definition}[Young order]
  Take two paths $\lambda, \kappa \in \LL(W, H)$ in height representation. We define a partial order, called the \emph{Young order}, where $\lambda \le \kappa$ iff $\lambda_a \le \kappa_a$ for each mode $a$.
\end{definition}

It is called the Young order because lattice paths can be seen as generalized Young diagrams when we fill the boxes below them.\footnote{Note the choice of orientation is such that lattice paths are nondecreasing sequences.} Then the order can be understood geometrically as the inclusion of Young diagrams; we show an example inequality below:
\ctikzfig{Young_order_geometric}
The Young order is well-known in the literature and has many properties useful to us later~\cite{stanley_enumerative_2012}.

We present an example of the order for $\LL(2,2)$ in \Cref{fig:LP-poset-Fock-space-2-3} which shows the order's Hasse diagram: a directed graph where vertices are elements of $\LL(2,2)$ and there is an arrow $\lambda \to \kappa$ iff $\lambda$~\emph{is covered by}~$\kappa$ (written $\lambda \lessdot \kappa$). The latter means that $\lambda < \kappa$ and there is no element between them. In the geometric picture, $\lambda \lessdot \kappa$ when the Young diagram of $\lambda$ has exactly one box less than $\kappa$.

\begin{figure}[t]
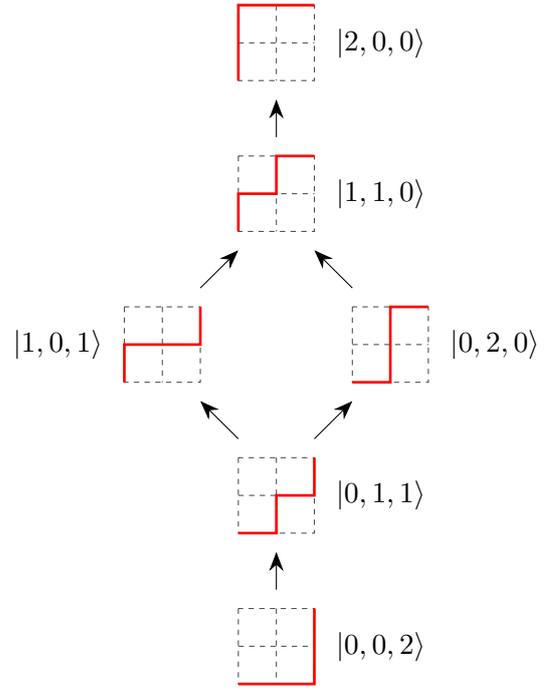

  \centering
  \ctikzfig{LP_Fock-space_2-3}
  \caption{The Hasse diagram of the set of paths $\LL(2,2)$, which corresponds to the bosonic Fock space $\Fs 23$, with the Young order. The arrows show the covering relation $\lambda \lessdot \kappa$. Next to each lattice path, we display the number state it represents.}
  \label{fig:LP-poset-Fock-space-2-3}
\end{figure}

\subsection{Representing measurements}
\label{app:lp-measurements}

The part of the formalism that encodes measurements uses \emph{intervals} in the Young poset, which in our paper are called \emph{skew diagrams} in analogy to skew Young diagrams. Recall that an interval is a subset of the poset defined by two paths $\nu \le \mu$:
\[ \mu/\nu \equiv [\nu, \mu] \deq \{ \lambda : \nu \le \lambda \le \mu \}. \]
When necessary, we write it vertically as $\frac\mu\nu$.

\subsubsection{Projection to skew diagrams}
\label{app:lp-projection}

Recall that to project to a subspace compatible with some measurement, say $n_a = x$, in \Cref{sec:projection}, we impose an S-shaped region onto the original downset $\downset\mu$. For the following steps, it is useful to always reason in terms of skew diagrams; note that the downset is the skew diagram $\mu/\bot$ where $\bot$ is the bottom element of the poset, i.e. the path $\lambda^\bot = (0,\dots,0,n)$.

The S-shape is parametrized by its starting height, controlling the number of photons in the modes to the left of the S, thus implementing the constraint $\lambda_{a-1} = q$ on the lattice paths $\lambda$ in the post-measurement subspace. It then takes a step of size $x$, encoding the number of photons found in mode $a$, implementing the requirement that $n_a = x$. Away from the S-shape, the skew diagram should stay as unchanged as possible from the original downset. A~convenient way to construct this is to define \emph{measurement helpers}, skew diagrams which only encode information about the measurement, and leave everything else as free as possible. We show an example in \Cref{fig:helper-diagram}, and we precisely define helper diagrams below. Then we take an intersection of the helper with the original skew diagram $\mu/\bot$ to obtain a projection.

\begin{definition}[measurement helper]
  \label{def:measurement-helper}
  To represent a subspace of the full Fock space $\gen{\LL(m-1,n)}_\C$ that has in total $\lambda_{a-1} = q$ photons in modes $\iicc{0, a-1}$ and~$x$ photons in mode~$a$, define a \emph{helper skew diagram} $E(q) = \varepsilon^q / \eta^q$ where the top path is
  \[
    \varepsilon^q
    = (\underbrace{q, \dots, q}_{a},
    q + x,
    \underbrace{n, \dots, n}_{m - a - 1})
  \]
  and the bottom is
  \[
    \eta^q
    =
    \begin{cases}
      (\underbrace{0, \dots, 0}_{a - 1}, q,
      \underbrace{q + x, \dots, q + x}_{m - a - 1}, n) & \text{if}\ a > 0, \\
      (\underbrace{q + x, \dots, q + x}_{m - a - 1}, n) & \text{if}\ a = 0.
    \end{cases}
  \]
  Note that if $a = 0$, then only $E(q=0)$ exists, so the above second case becomes $\eta^0 = (x,\dots,x,n)$. Note also that the maximum $q$ is $q_{\max} = n-x$.
\end{definition}

The notation $E(q)$ does not specify the mode $a$, nor the number of photons $x$, however, in the following text, this will always be clear from the context. 

\begin{figure}[t]
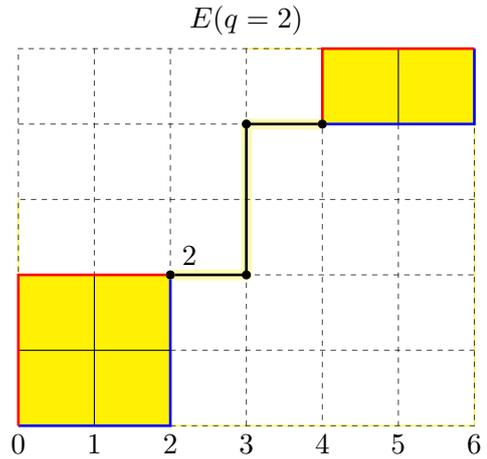

  \ctikzfig{helper_diagram}
  \caption{Helper diagram $E(q = 2)$ which represents the subspace of the full Fock space $\gen{\LL(m-1, n)}_\C$ (with $m=7$ and $n=5$) of those basis states that have $q=2$ photons in modes $\iicc{0,2}$, represented by the starting position of the S-shaped region, and $x = 2$ photons in mode $3$, represented by the height of the step in the S-shape. No other constraints are implemented in this skew diagram.}
  \label{fig:helper-diagram}
\end{figure}

The above definition is the first step to formally define the projection step in terms of lattice paths, i.e. in a way that can be written as a formula (and implemented in code). We use it to represent the subset of the basis of a full Fock space that remains after a given measurement. The next step is to take its \emph{intersection} with the diagram $\mu/\bot$ representing the state space of our system before measurement to obtain those basis states that agree with both. In order to obtain the intersection of skew diagrams, we first need some more theory of the Young order:

The poset $\LL(m-1, n)$ is an order-theoretic \emph{lattice} which means that it contains all finite meets and joins~\cite{stanley_enumerative_2012}. Recall that a \emph{meet} (also called \emph{supremum} or \emph{least upper bound}) of $\lambda$ and $\kappa$ is the unique least element $\lambda \wedge \kappa$ such that $\lambda \le \lambda \wedge \kappa$ and $\kappa \le \lambda \wedge \kappa$. Dually, the \emph{join} (also \emph{infimum} or \emph{greatest lower bound}) is the unique greatest element $x \vee y$ such that both $x$ and $y$ are greater.

\begin{lemma}[meet and join in Young order]
  Suppose $\lambda, \kappa \in \LL(m-1, n)$. Then their meet and join have the following components:
  \begin{align*}
    (\lambda \wedge \kappa)_a
    & = \min\{ \lambda_a, \kappa_a \}, \\
    (\lambda \vee \kappa)_a
    & = \max\{ \lambda_a, \kappa_a \}.
  \end{align*}
\end{lemma}
\begin{proof}
  Recall that in the Young order $\alpha \le \beta$ iff $\alpha_a \le \beta_a$ for all $a$, making it a product order of several copies of the linear order on integers. This means the meet and join are componentwise the meet and join of integers, which are the minimum and maximum, respectively. Note that this is compatible with the fact that the height representation of a lattice path must be a nondecreasing sequence: a component $a$ is a meet (resp. join), and so is $a+1$, etc., so the nondecreasing constraint is satisfied.
\end{proof}

\begin{lemma}[intersection of skew diagrams]
  \label{lem:intersection-skew-diagrams}
  Suppose $\alpha/\beta$ and $\gamma/\delta$ are skew diagrams in $\LL(m-1, n)$. Their intersection is either empty or equal to the following:
  \[
    \frac\alpha\beta \cap \frac\gamma\delta
    =
    \frac{ \alpha \wedge \gamma }{ \beta \vee \delta }.
  \]
\end{lemma}
\begin{proof}
  The intersection must contain exactly the paths $\lambda$ such that $\beta \le \lambda \le \alpha$ and simultaneously $\delta \le \lambda \le \gamma$. It is possible that no such $\lambda$ exists, in which case the intersection is empty. If the intersection is nonempty, then it is clearly an interval (skew diagram). The greatest $\lambda$ that satisfies both $\lambda \le \alpha$ and $\lambda \le \gamma$ is the meet $\alpha \wedge \gamma$. Dually, the least $\lambda$ is the join $\beta \vee \delta$. Note that this relies on $\LL(m-1, n)$ being an order lattice.
\end{proof}

With the above tools, we can characterize the state spaces after the measurement projection:

\begin{lemma}
  \label{lem:projected-diagrams}
  The family of skew diagrams obtained from $\mu/\nu$ after the projection implementing the measurement of $x$ photons in mode $a$ consists of skew diagrams parametrized by the number of photons $q$ in modes $\iicc{0, a-1}$ written as follows:
  \[
    \frac{\mu^q}{\nu^q} = \frac\mu\nu \cap E(q)
    = \frac{\mu \wedge \varepsilon^q}{\nu \vee \eta^q}.
  \]
\end{lemma}
\begin{proof}
  The intersection selects those lattice paths $\lambda$ in the original state space $\mu/\nu$ which also satisfy $\lambda_{a-1} = q$ and $n^\lambda_a = x$, the latter constraints imposed by the measurement helper $E(q)$. The second equality follows from \Cref{lem:intersection-skew-diagrams}.
\end{proof}

Note that the values of $q$ which yield nonempty intersections $\mu^q/\nu^q$ come from an interval $\iicc{q_{\min}, q_{\max}}$ whose endpoints depends on $\mu$ and $\nu$ and $x$. Notably, in our case of $\nu=\bot$, these are $q_{\min} = 0$ and $q_{\max} = \min\{\mu_{a-1}, \mu_a - x\}$.

\subsubsection{Contraction}
\label{app:lp-contraction}

Having the projection, we now formalize the next step, the contraction of the projected skew diagrams from \Cref{sec:contraction}. The contraction step is required to make the merge step in the next section work. Recall that the contraction removes a mode and all the photons contained therein.

\begin{definition}[contraction]
  In a poset $\LL(m-1, n)$, measuring $x$ photons in mode $a$ and then deleting it is represented by the function
  \[
    C_a^x : \bigcup_q E(q) \to \LL(m-2, n-x)
  \]
  that sends
  \[
    \begin{tikzcd}[ampersand replacement=\&]
      { (n_0, \dots, n_{a-1}, x, n_{a+1}, \dots, n_{m-1})} \\
      { (n_0, \dots, n_{a-1}, \phantom{x,\ } n_{a+1}, \dots, n_{m-1})}
      \arrow[maps to, from=1-1, to=2-1]
    \end{tikzcd}
  \]
  deleting the component $a$ in step representation and preserving the others. The $q$ in $\bigcup_qE(q)$ runs over all allowed values $\iicc{0, n-x}$.
  The domain is the subset of $\LL(m-1, n)$ where the contraction makes sense.
\end{definition}

Recalling \cref{eq:n-to-lambda}, the height representation is transformed as follows:
\begin{equation}
  \label{eq:contracted-lambda}
  (C_a^x \lambda)_b =
  \begin{cases}
    \lambda_b & \text{if}\ b \le a-1, \\
    \lambda_{b+1} - x & \text{if}\ b \ge a.
  \end{cases}
\end{equation}

What follows are two technical lemmas showing that the contraction is well-behaved with respect to the Young order: it is an order isomorphism, which means it is a bijection of sets and preserves the order structure of those sets. This is important because it is this order structure that allows us to efficiently count basis states.

\begin{lemma}[contraction is monotone]
  \label{lem:contraction-monotone}
  The contraction function $C_a^x : \bigcup_qE(q) \to \LL(m-2, n-x)$ preserves the Young order.
\end{lemma}
\begin{proof}
  Take two paths $\lambda, \kappa \in \bigcup_qE(q)$, and suppose $\lambda \le \kappa$. Equivalently $\lambda_b \le \kappa_b$ for each $b$. The contraction deletes mode $a$ and all the photons contained therein. Following \Cref{eq:contracted-lambda}, for modes $b \le a-1$, we have \[(C_a^x\lambda)_b = \lambda_b \le \kappa_b = (C_a^x\kappa)_b.\] Similarly, for modes $b \ge a$, we have \[(C_a^x\lambda)_b = \lambda_{b-1} - x \le \kappa_{b-1} - x = (C_a^x\kappa)_b.\] We conclude that $C_a^x\lambda \le C_a^x\kappa$.
\end{proof}

\begin{lemma}[contraction is isomorphism]
  \label{lem:contraction-iso}
  The contraction function $C_a^x$ is an order isomorphism, i.e. a monotone bijection $\bigcup_qE(q) \isoto \LL(m-2, n-x)$ whose inverse is also monotone.
\end{lemma}
\begin{proof}
  The function $C_a^x$ is clearly injective: it forgets the $a$-th component of a lattice path in the step representation, but that component is, by construction, always $n_a = x$ for the paths contained in the domain $\bigcup_qE(q)$. Thus it forgets no information, and is clearly invertible, at least from its image, written $\bigcup_q\widetilde E(q)$, where each $\widetilde E(q) = C_a^x(E(q))$. Writing its inverse $(C_a^x)^{-1} : \bigcup_q\widetilde E(q) \to \bigcup_qE(q)$ is easy in the step representation: simply insert a component of value $x$ at position $a$ to undo the action of $C_a^x$. We know from \Cref{lem:contraction-monotone} that $C_a^x$ preserves order, and by a similar argument, we see that $(C_a^x)^{-1}$ does as well.

  What remains to be shown is that $C_a^x$ is surjective, i.e. that $\bigcup_q\widetilde E(q) = \LL(m-2, n-x)$. Take an arbitrary path $\lambda \in \LL(m-2, n-x)$. We show that we can act by the inverse $(C_a^x)^{-1}$. Inserting a component $x$ at position $a$ into $\vec n^\lambda$ is clearly possible and sends $\lambda$ to $E(q = \lambda_{a-1})$. Finally, we note that $E(q=\lambda_{a-1})$ is indeed a part of the domain $\bigcup_qE(q)$: recall that $q$ ranges over $\iicc{0,n-x}$, as does $\lambda_{a-1}$.
  We conclude that $C_a^x$ is a bijection of sets, and an order isomorphism.
\end{proof}

\subsubsection{Merge}
\label{app:lp-merge}

The final step is to merge the contracted skew diagrams into a single skew diagram, and indeed a downset, containing exactly the correct paths. We fill in the details omitted in \Cref{sec:merge}. Skew diagrams are intervals in the poset, and merging them means taking their union as sets:
\[
  \bigcup_q C_a^x \left(\frac{\mu^q}{\nu^q}\right).
\]

It is not entirely trivial to see that such a union is also an interval. Indeed, this is not the case when there is a gap between the intervals, i.e. there exist elements between them. To see the intuition, we use an example of the real line with its usual order: Suppose we have two intervals $[a,b]$ and $[c,d]$ s.t. $b \le c$. If $b < c$, then there exists some $x \in (b,c)$ that is not contained in either interval, nor in their union $[a,b] \cup [c,d]$. This means that the union is not itself an interval.

One can easily find examples to show that, without contraction, the union $\bigcup_qE(q)$ (and therefore $\bigcup_q \frac{\mu^q}{\nu^q}$) is not generally an interval. However, the contraction $C_a^x$ removes the gaps between consecutive helper diagrams, i.e. intervals $\widetilde E(q)$ and \mbox{$\widetilde E(q+1)$}, where we recall the notation $\widetilde E(q) \equiv C_a^x(E(q))$. The strategy to show this is to study the boundary between a contracted helper $\widetilde E(q+1)$ and a hypothetical gap between $\widetilde E(q)$ and $\widetilde E(q+1)$, showing that such gap cannot exist.

\paragraph{Boundaries of intervals}

To do this, define a \emph{lower boundary} of a subset~$S$ of a poset~$P$ as
\[ \partial S = \{ (\lambda, \kappa) : \kappa \in S, \lambda \in P \setminus S, \lambda \lessdot \kappa \}. \]
If $S$ is a contiguous region in the Hasse diagram, then $\partial S$ can be seen as a cut between $S$ and the region outside (and lower than) $S$.
We can dually define an \emph{upper boundary}, but the reasoning will be analogous, so it is enough to focus on $\partial S$.

The strategy is to show that for a contracted interval, the lower point~$\lambda$ of the boundary belongs to another contracted interval. We start with a general statement:

\begin{lemma}
  \label{lem:boundary-of-interval}
  Let $\alpha/\beta$ be some skew diagram, and let $(\lambda, \kappa) \in \partial(\alpha/\beta)$, then there exists a component~$b$ such that
  \[
    \lambda_c =
    \begin{cases}
      \beta_b - 1 & \text{if}\ c = b, \\
      \kappa_c & \text{if}\ c \ne b.
    \end{cases}
  \]
\end{lemma}
\begin{proof}
  Recall that if $\lambda \lessdot \kappa$, then geometrically $\lambda$ has one fewer box than $\kappa$ as a Young diagram. In terms of the height representation, there must be exactly one component, here $b$, where they differ: $\lambda_b = \kappa_b - 1$.
  By definition of the boundary, \mbox{$\kappa \in \alpha/\beta$}, thus $\kappa_c \in \iicc{\beta_c, \alpha_c}$ for each $c$. On the other hand, $\lambda \notin \alpha/\beta$, so there must exist at least one $c$ for which $\lambda_c \notin \iicc{\beta_c, \alpha_c}$, which implies $\lambda_c < \beta_c$ since $\lambda \le \kappa$.
  Since $\lambda$ and $\kappa$ must agree on all but one element, the component $c$ where this happens is $c=b$.
  Together, this implies that $\kappa_b = \beta_b$ and $\lambda_b = \beta_b-1$.
\end{proof}

\begin{lemma}
  \label{lemma:boundary-in-helper}
  Let $(\lambda, \kappa) \in \partial C_a^x(E(q+1))$. Then the value $b$ from \Cref{lem:boundary-of-interval} is $b=a-1$.
\end{lemma}
\begin{proof}
  First, we show that $b \ge a-1$ which follows from \Cref{lem:boundary-of-interval}: we have $\lambda_b = C_a^x(\eta^{q+1})_b - 1$, and by \Cref{def:measurement-helper}, for $c \in \iicc{0, a-2}$, the bottom path has components $\eta^{q+1}_c = C_a^x(\eta^{q+1})_c = 0$. The value $\lambda_b$ must be nonnegative, and thus $b \ge a-1$.

  Next, we show that $b \le a-1$. By case-by-case analysis, we show that $(\lambda, \kappa)$ in the boundary of a contracted helper diagram can not exist otherwise. We display the cases in the table below, omitting their individual derivation which is simple but tedious. Each case shows the allowed values of $\lambda_c$ depending on the choice of $b$ and $c \ne b$.
  \begin{center}
    \begin{tabular}{ |r|c|c| }
      \hline
      & $c < b$ & $c > b$ \\
      \hline
      $c \in \iicc{0, a-2}$ & $\iicc{0,q}$ & $\{ q, q+1 \}$ \\
      $c = a-1$ & $\varnothing$ & $\{ q+1 \}$ \\
      $c \in \iicc{a, m-3}$ & $\varnothing$ & $\iicc{q+1, n-x}$ \\
      $c = m-2$ & $\varnothing$ & $\{ n-x \}$ \\
      \hline
    \end{tabular}
  \end{center}
  Cells which violate order requirements are shown by $\varnothing$ to signify that they represent cases with no possible value for~$\lambda_c$. Conclude that the only way to have a value of $\lambda_c$ for all $c$ is if $b \le a-1$.
\end{proof}

\begin{lemma}
  \label{lem:boundary-connects-intervals}
  Let $(\lambda, \kappa) \in \partial C_a^x(E(q+1))$. Then $\lambda \in C_a^x(E(q))$.
\end{lemma}
\begin{proof}
  This follows from \Cref{def:measurement-helper} and \Cref{lem:boundary-of-interval,lemma:boundary-in-helper}. We know that
  \[
    \lambda_{a-1} = \widetilde\eta^{q+1}_{a-1} - 1 = q = \widetilde\eta^q_{a-1},
  \]
  where we denote $\widetilde\alpha \equiv C_a^x(\alpha)$ for paths $\alpha$. For all $c \ne a-1$, we have $\lambda_c = \kappa_c \in \iicc{\widetilde\eta^{q+1}_c, \widetilde\varepsilon^{q+1}_c}$. For $c < a-1$, this is $\iicc{0, q}$, and for $c > a-1$, it is $\iicc{q+1, n-x}$. We conclude that $\lambda \in C_a^x(E(q))$.
\end{proof}

\begin{lemma}[union of consecutive intervals]
  \label{lem:union-of-cont-int}
  The union of consecutive contracted intervals $\widetilde E(q)$ and $\widetilde E(q+1)$ is
  \[
    \widetilde E(q) \cup \widetilde E(q+1) = \frac{\widetilde\varepsilon^{q+1}}{\widetilde\eta^q}.
  \]
\end{lemma}
\begin{proof}
  This requires three facts, two of which are immediate:
  First, the set $\widetilde E(q) \cup \widetilde E(q+1)$ has a maximum element $\widetilde\varepsilon^{q+1}$, i.e. the top element of the interval $\widetilde E(q+1)$.
  Similarly, it has a minimum $\widetilde\eta^{q}$, i.e. the bottom element of $\widetilde E(q)$. This follows from the construction of (contracted) helper diagrams.

  The final step is to show that there are no elements between the intervals $\widetilde E(q)$ and \mbox{$\widetilde E(q+1)$}. If $\lambda$ is between them, there exist $\alpha \in \widetilde E(q)$ and $\beta \in \widetilde E(q+1)$ s.t. $\alpha < \lambda < \beta$. It is possible to find such $\lambda$ that is adjacent to $\widetilde E(q+1)$ in the Hasse diagram, i.e. it is covered by a $\kappa \in \widetilde E(q+1)$. This is equivalent to saying $(\lambda, \kappa) \in \partial \widetilde E(q+1)$. From \Cref{lem:boundary-connects-intervals}, it follows that $\lambda \in \widetilde E(q)$. We conclude that there is no gap between the two consecutive contracted helper intervals, and thus their union is the interval $\widetilde\varepsilon^{q+1}/\widetilde\eta^q$.
\end{proof}

\subsubsection{Full description of measurement}

The result of \Cref{lem:union-of-cont-int} clearly extends to the whole union
\[
  \bigcup_q \widetilde E(q)
  = \frac{\widetilde\varepsilon^{q_{\max}}}{\widetilde\eta^{q_{\min}}}
  = \frac{\widetilde\varepsilon^{n-x}}{\widetilde\eta^{0}}.
\]
Furthermore, it extends to the union of contracted intersections $\bigcup_q C_a^x(\mu/\nu \cap E(q))$: each intersection is taken with the same skew diagram $\mu/\nu$, and the effect is simply a further restriction on modes away from $a-1$.

Finally, having proved all that is needed for \Cref{thm:lp-measurement} in the preceding sections, we recall its statement below:

\noindent
{\normalfont\bfseries Theorem\ \ref*{thm:lp-measurement}\ {\normalfont (measurement)}.}
Starting with a cumulative space $\gen{\downset\mu}$, the state space left after measuring $x$ photons in mode $a$ is a cumulative space $\negen{\downset{\mu^{(n_a = x)}}}$ where the maximal path $\mu^{(n_a = x)}$ is the contraction of $\mu^{q_{\max}}$. The latter is the maximum path of the skew diagram corresponding to $q = q_{\max} = \min\{ \mu_{a-1}, \mu_a - x \}$.

\subsection{Coupling new modes in initial $L_1$}
\label{sec:coupling-new-modes}

In \Cref{sec:tracking-state-space}, we mention a subtlety in applying beamsplitters from the initial loop $L_1$. The issue is that the initial loop couples a PCS to a new, previously unseen mode, with new photons; thus the rule from \Cref{thm:PCS-BS-evolution} does not apply.

Though it is likely not possible in general to couple an additional mode to an arbitrary mode of a PCS, it is possible in the special case of coupling the maximum mode of a PCS to a new mode via a beamsplitter of the initial $L_1$. Suppose we start with $m$ modes and $n$ photons. All we need to do is expand the lattice rectangle width by one mode, height by the number $x$ of new photons in the added mode, and to set the height of $\mu$ in the final two modes to the new total number of photons $n+x$:
\[
  (\mu_0, \dots, \mu_{m-2}, n) \mapsto (\mu_0, \dots, \mu_{m-2}, n+x, n+x).
\]

\subsection{Counting lattice paths}
\label{app:counting-lp}

\begin{figure}[t]
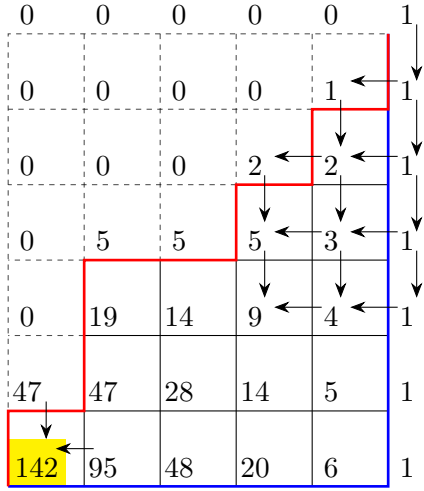

  \ctikzfig{dp_lp_counting4}
  \caption{A schematic explanation of the counting algorithm for an example skew diagram $\mu/\nu$ (in fact, a downset, i.e. with $\nu$ the bottom path). To each point $(x,y)$ in the lattice rectangle, we associate a number $A_{x,y}$. This is the number of allowed lattice paths within a part of the the skew diagram that start at point $(x,y)$ and terminate, as usual, in the top right corner $(W,H)$. The array $A$ is what \protect\Cref{alg:dp-lp-counting} computes. The total number of paths in the diagram is found at point $(0,0)$, i.e. $|\mu/\nu| = A_{0,0}$.}
  \label{fig:dp-lp-counting}
\end{figure}

\begin{myalgorithm}{%
    Pseudocode of the dynamic programming algorithm to count the lattice paths in a given skew diagram. For notational simplicity, we allow the indices $i,j$ of $A_{i,j}$ to be out of bounds ($i > m-1$ or $j > n$) and we assume $A_{i,j} = 0$ in those cases.
  }{\label{fig:pseudocode-dp-lp-counting}}
  \KwIn{Lattice paths $\nu \le \mu$ identifying a skew diagram $\mu/\nu$.}
  \KwOut{Cardinality $|\mu/\nu|$.}

  $m \gets$ length of $\mu$; $n \gets \mu_{m-1}$\;

  \lIf{$m = 1$ or $n = 0$}{
    \KwRet{0}
  }

  Initialize array $A = (A_{i,j})_{i,j}$ with all zeros, for $i=0,\dots,m-1$ and $j = 0, \dots, n$ \;

  $A_{m-1, n} \gets 1$\;

  \For{$x \gets m-1, m-2, \dots, 0$}{
    \lIf{$x = 0$}{$y_{\min} \gets 0$}
    \lElse{$y_{\min} \gets \nu_{x-1}$}

    \For{$y \gets \mu_x, \mu_x - 1, \dots, y_{\min}$}{
      $A_{x,y} \gets A_{x+1, y} + A_{x, y+1}$ where $A_{i,j} = 0$ if $i$ or $j$ are out of bounds\;
    }
  }

  \KwRet{$A_{0,0}$}

  \caption{count\_interval}
  \label{alg:dp-lp-counting}
\end{myalgorithm}

We give a dynamic programming algorithm to efficiently count a downset $\downset\mu$ or a skew diagram $\mu/\nu$. To formulate it, we recall that a lattice path can be understood as a sequence of steps right $(+1,0)$ or up $(0,+1)$.

Note that given a point $(x,y)$ in the lattice rectangle $\iicc{0,W} \times \iicc{0,H}$, the number of paths passing through this point is
\[ |\LL(x,y)| \cdot |\LL(W-x, H-y)|. \]
This is the number of paths that reach $(x,y)$, times the number of paths that start at $(x,y)$ and reach $(W,H)$. To count the paths in $\mu/\nu$, we restrict the above sets to regions allowed by the skew diagram.

Together, the two ideas above give us an efficient algorithm to compute the number of lattice paths in a skew diagram $\mu/\nu$; see example in \Cref{fig:dp-lp-counting} and pseudocode \Cref{alg:dp-lp-counting} in \Cref{fig:pseudocode-dp-lp-counting}.

To each point $(x,y)$ of the lattice rectangle $\iicc{0,W} \times \iicc{0,H}$, we associate an integer $A_{x,y}$ that is the number of paths starting at that point, ending at $(W,H)$, and staying within the region $\mu/\nu$. These form an array $A$, shown in \Cref{fig:dp-lp-counting}. From $(x,y)$, a path can go up to $(x,y+1)$ or right to $(x+1, y)$, and at each of those points, we have a corresponding value of $A$ which counts the (disjoint) set of paths starting there. This allows us to count paths starting at $(x,y)$:
\[
  A_{x,y} = A_{x+1, y} + A_{x, y+1}.
\]
Above, by abuse of notation, we assume that $A_{x,y} = 0$ if $(x,y) \notin \iicc{0,W} \times \iicc{0,H}$. Points outside of $\mu/\nu$ also get $A_{x,y} = 0$. Finally, we conventionally define $A_{W,H} = 1$: there is a single (empty) path from $(W,H)$ to itself.

To evaluate all entries of $A$, we start with $A_{W,H} = 1$ and progress through the array towards $(0,0)$, evaluating an entry once its dependencies are available, as shown in \Cref{fig:dp-lp-counting} and \Cref{alg:dp-lp-counting}. The final value is
$A_{0,0} = |\mu/\nu|$.

\section{Theoretical complexity of the progressive sampling algorithm}
\label{sec:theor-compl-progr}

In \Cref{sec:comp-with-other}, we compared the asymptotics of our Algorithm~PS to other algorithms in the literature. Here, we provide the derivation. This will concern our reference implementation of the algorithm, found in the companion Git repository.\footref{fn:git}

Since the algorithm is dynamic in nature, in that the computations for component $P_{a+1}$ depends on the number of photons measured after component $P_a$, it would be difficult to estimate its theoretical time complexity in the general case.
Thus here we focus on one kind of circuit of interest, the power-law architecture with loop lengths $\vec\ell = (1, \ell, \ell^2, \dots \ell^{\Lambda-1})$. Note that in this case, the number of modes per components of the progressive decomposition, given by \eqref{eq:relevant-modes}, is:
\begin{subequations}
  \begin{align}
    \nonumber
    R & = 1 + \sum_{i=1}^\Lambda \ell^{i-1} \\
    \label{eq:rel-modes-power-law-real}
      & =
        \begin{cases}
          1 + \Lambda & \text{if}\ \ell = 1,\\
          1 + \frac{\ell^{\Lambda} - 1}{\ell - 1} & \text{if}\ \ell > 1,
        \end{cases} \\
    \label{eq:rel-modes-power-asymp}
      & =
        \begin{cases}
          \Theta(\Lambda) & \text{if}\ \ell = 1,\\
          \Theta(\ell^{\Lambda-1}) & \text{if}\ \ell > 1.
        \end{cases}
  \end{align}
\end{subequations}
We take a total number of modes $m \ge R$ (even $m \gg R$ generally), and we take an input state $\ket{\vec n} = \ket{1,0,1,0,\dots}$ that we used in \Cref{sec:quantifying-complexity}. The analysis can be easily generalized to input $\ket{\vec n} = \ket{n_0, n_1, \dots, n_{m-1}}$ with $n_a \le 1$ for all modes $a$, by bounding the photon numbers by $R$ instead of $R/2$ in the following.

\subsection{Time to compute transition amplitudes}

\subsubsection{Single beamsplitter}

In our reference implementation, we use the following equation to compute the transition amplitudes for a beamsplitter $B_{ab}(\theta)$ between modes $a \ne b$ with beamsplitter angle $\theta$:
\begin{widetext}
  \begin{equation}
    \label{eq:BS-amplitude}
    [ B_{ab}(\theta) ]_{n_a, n_b}^{N_a, N_b} = \sqrt{n_a!\, n_b!\, N_a!\, N_b!} \sum_{k = \max\{ 0, N_b - n_a \}}^{\min\{ n_a, N_a \}} \frac{(-1)^{n_a - k + 1} (\cos \theta)^{n_b + k - l} (\sin \theta)^{n_a -k + l}}{k!\, (n_a - k)!\, l!\, (n_b - l)!},
  \end{equation}
\end{widetext}
where $k$ represents the number of photons starting in mode $a$ and ending in the same mode $a$, and where $l = N_a - k$, the number of photons transitioning from $b$ to $a$. Furthermore $n_a,n_b$ are the input, and $N_a, N_b$ the output numbers of \mbox{photons}.\footnote{To get (\ref{eq:BS-amplitude}) from (\ref{eq:BS-on-modes}), represent photons in modes $a,b$ by creation operators $\hat a_a^\dagger, \hat a_b^\dagger$ that transform  as $\hat a_c^\dagger \mapsto \sum_{d \in \{ a, b \}} B(\theta)_{dc} \hat a^\dagger_d$, and extend this to $(\hat a^\dagger_a)^{n_a} (\hat a^\dagger_b)^{n_b}$.}

Assuming $\Theta(1)$ time arithmetic\footnote{including using $q! = \exp\ln q! = \exp \sum_{i=1}^q \log i$ to efficiently compute factorials, as well as precomputing the first several values of $\log q!$}, the complexity of computing one transition amplitude from (\ref{eq:BS-amplitude}) is just the number of terms in the summation, i.e. the number of possible values of $k \in \iicc{\max\{0, N_b - na\}, \min\{ n_a, N_a \}}$. We show the size of this interval in the table below:
\begin{center}
  \begin{tabular}{|c|c c|}
    \hline
    & $N_b \le n_a$ & $N_b > n_a$ \\
    \hline
    $n_a \le N_a$ & $n_a + 1$ & $2 n_a - N_b + 1$ \\
    $n_a > N_a$ & $N_a + 1$ & $n_a + N_a - N_b + 1$ \\
    \hline
  \end{tabular}
\end{center}
Note that $\tilde n \deq n_a + n_b = N_a + N_b$, and $N_a = k + l$. It is clear that in each case the number can be bounded by $2\tilde n + 1$.

\paragraph{Simplification: beamsplitters see all}
To compute evolution of the state though a beamsplitter, we need to apply (\ref{eq:BS-amplitude}) for all possible values of $n_a, n_b, N_a, N_b$. We can exactly predict the constraints on these using the lattice path formalism from \Cref{sec:a-new-formalism}, however, this would be too complicated for an asymptotic analysis. Instead, we simplify and assume that \emph{every beamsplitter of a component $P_c$ may interact with all the photons in that component.} This will lead to an upper bound, because most beamsplitters will see much fewer photons. The number of photons in the component $P_c$ is denoted $n(P_c)$. Thus we bound $\tilde n = n_a + n_b \le n(P_c)$.

There are $\tilde n + 1$ possible way to distribute $\tilde n$ photons between $N_a, N_b$. Likewise, there are $\tilde n + 1$ ways to do this for the input numbers of photons $n_a, n_b$, but we shall count these in a later step when we reason about the size of the statevector to be evolved. Thus for each beamsplitter we need to compute up to $\tilde n + 1 \le n(P_c) + 1$ amplitudes, and together with the above number of terms per amplitude, we get up to $(n(P_c) + 1) \cdot (2 n(P_c) + 1)$ operations.

\subsubsection{All beamsplitters in the circuit}

We now count how many beamsplitters are contained in the circuit so that we can estimate the time needed to compute the transformation of the state through them all. Each loop of length $\ell'$ generates a sequence of $m - \ell'$ beamsplitters: there is a beamsplitter starting at each mode, except the $\ell'$ last ones. For simplicity, we upper bound this by $m$, noting also that generally $m \gg R > \ell'$. There are $\Lambda$ loops, so the number of beamsplitters in the circuit is bounded by $m\Lambda$.

Note that the reason we can count all beamsplitters in the circuit, instead of worrying about individual components, is a later simplification that assumes all components behave similarly.

Observe that the complexity depends on the number of photons in mode $P_c$. This complicates the analysis, because these numbers depend on the actual SPBS marginal probabilities. We will again simplify the analysis by taking an appropriate assumption.
Note first that
\[
  n(P_{c+1}) = n(P_c) - n'_c + n_{c + 1 + R},
\]
which corresponds to removing the $n'_c$ measured photons coming from the previous component $P_c$, and adding $n_{c + 1 + R}$ photons in the new input mode.

\paragraph{Simplification: number of photons roughly constant}
Without any more information, it is reasonable to assume that on average, the number of photons in a component will stay roughly the same, i.e. $n(P_{c+1}) \approx n(P_c)$. This takes into account that we add photons, but not in all new input modes, and we remove photons by measurement, but we cannot predict how many. At its core, this is an entropy argument: we cannot assume any more information.

With this simplification, the number of photons in each component becomes $n(P_c) \approx n(P_0) = \lceil R/2 \rceil \approx R/2$ (or generally $\le R$ for an input with $n_a \le 1$). Then the number of operations needed to compute amplitudes becomes the number of operations per beamsplitter multiplied by the number of beamsplitters (still ignoring the state space):
\begin{align}
  \nonumber
  T^{\mathrm{amp}}
  & = \mathcal{O}\left( m \Lambda \cdot \left( \frac R2 + 1 \right) \cdot (R + 1) \right) \\
  & = \mathcal{O}(m \Lambda R^2).
\end{align}
Recall again that $R = 1 + \Lambda$ if $\ell = 1$, and $R = \mathcal{O}(\ell^{\Lambda - 1})$ if $\ell > 1$. Then the time to compute the amplitudes of beamsplitters needed to evolve the state is:
\begin{equation}
  \label{eq:total-time-ev}
  T^{\mathrm{amp}} =
  \begin{cases}
    \mathcal{O}(m \Lambda^3) & \text{if}\ \ell = 1, \\
    \mathcal{O}(m \Lambda \ell^{2(\Lambda - 1)}) & \text{if}\ \ell > 1.
  \end{cases}
\end{equation}

\subsection{Effect of state in memory}

Finally, we recall that so far we only computed the time to compute the transition amplitudes. In order to evolve the state vector, we have to multiply its coordinate components by these values, which essentially means iterating over these components.

The number of statevector components, i.e. the dimension of the reachable state space, at each step, is given by the lattice paths. However, this would again make asymptotic analysis difficult. Instead, we extend the simplification we had already taken before: each beamsplitter can see all the photons in a progressive component $P_c$, and so we further assume that the Fock space at the input and output of each beamsplitter is the same as the final space at the end of the progressive component. Having assumed that $n(P_c) \approx n(P_0) \approx R/2$, we can bound the dimension of the relevant space by:
\[
  \binom{R + n(P_0) - 1}{R}
  \approx \binom{\frac32 R - 1}{R}
\]
We simplify this using Stirling's approximation
\[ \ln x! = x \ln x - x + \mathcal{O}(x) \]
which gives
\[ \binom{x}{qx} \sim x \Bigl( -q \ln q - (1-q) \ln (1-q) \Bigr) \]
for a large $x$ and $q \in (0,1)$. For us $x = \frac32 R - 1$ and $q \to \frac23$ as $R$ grows, and thus
\begin{equation}
  \label{eq:stirling-3/2}
  \binom{\frac32 R - 1}{R} = \mathcal{O}\left( \left( \frac{\sqrt{27}}{2} \right)^R \right) = \mathcal{O}( 2.6^R )
\end{equation}
where we take an upper bound $\sqrt{27}/2 = 2.598\dots < 2.6$.

Finally, noting that the above is independent of the component $P_c$, we obtain a bound on the total running time by multiplying the above with $T^{\mathrm{amp}}$. We recall again that $R = 1 + \Lambda$ if $\ell = 1$, and $R = \mathcal{O}(\ell^{\Lambda - 1})$ if $\ell > 1$. Thus the running time is:
\begin{equation}
  \label{eq:runtime-final}
  T =
  \begin{cases}
    \mathcal{O}(m \Lambda^3 2.6^\Lambda) & \text{if}\ \ell = 1, \\
    \mathcal{O}(m \Lambda \ell^{2(\Lambda - 1)} 2.6^{\ell^{\Lambda - 1}}) & \text{if}\ \ell > 1.
  \end{cases}
\end{equation}

\subsubsection{Measurement subsumed}

In the end, we note that we did not mention the measurement step that occurs after every progressive component $P_c$. This consists of iterating over the components of the state vector, summing up the marginal probabilities of sampling a particular photon number using a hash table, and then taking a sample from the resulting distribution. This can be done in time $\mathcal{O}(n(P_c) 2.6^R) = \mathcal{O}(R 2.6^R)$ that is already subsumed by the bound in (\ref{eq:runtime-final}).

\bibliographystyle{quantum}
\bibliography{bibliography}

\end{document}